\PassOptionsToPackage{final}{graphicx}
\documentclass[autoref,UKenglish,final]{lmcs}
\pdfoutput=1

\usepackage{lastpage}
\lmcsdoi{22}{1}{9}
\lmcsheading{}{\pageref{LastPage}}{}{}%
{Mar.~04,~2024}{Feb.~16,~2026}{}

\usepackage[breaklinks=true,final]{hyperref}    

\usepackage{mathrsfs}

\newenvironment{theorem}{\begin{thm}}{\end{thm}}
\newenvironment{corollary}{\begin{cor}}{\end{cor}}
\newenvironment{lemma}{\begin{lem}}{\end{lem}}
\newenvironment{proposition}{\begin{prop}}{\end{prop}}
\newenvironment{assumption}{\begin{asm}}{\end{asm}}

\newenvironment{remark}{\begin{rem}}{\end{rem}}
\newenvironment{example}{\begin{exa}}{\end{exa}}
\newenvironment{definition}{\begin{defi}}{\end{defi}}

\usepackage[utf8]{inputenc}

\usepackage{adjustbox}

\usepackage{enumitem} 
\setlist{itemsep=0ex}

\def\dotminus{\mathbin{\ooalign{\hss\raise1ex\hbox{.}\hss\cr
  \mathsurround=0pt$-$}}}

\input{catprog}

\renewcommand{\case}[3]{\mfix{case}{\mathbin{}#1}{of}{#2}{\kern-1pt;}{\mathbin{}#3}}
\renewcommand{\inj}{\oname{inj}}

\newcommand{\trans}[1]{\mathrel{\raisebox{-1.0pt}{$\xrightarrow{\;\smash{\raisebox{4.75pt}{\makebox(9,0)[t]{\scriptsize $#1$}}\;}}$}}}

\newcommand{\docase}[3]{\mfix{docase}{\mathbin{}#1}{of}{#2}{\kern-1pt;}{\mathbin{}#3}}

\renewcommand{\iobj}{0}  

\usepackage{etex}
\usepackage{amssymb,stmaryrd} 
\usepackage{xcolor} 
\usepackage{wasysym}
\usepackage{stackengine}
\usepackage{bbold}

\newcommand{\real}{\mathbb{R}}
\newcommand{\realp}{\real_{\geq 0}}
\newcommand{\realsp}{\real_{> 0}}
\newcommand{\realpe}{\bar\real_{\geq 0}}

\usepackage{tikz}
\usepackage{wrapfig}

\usetikzlibrary{arrows}

\tikzset{cong/.style={draw=none,edge node={node [sloped, allow upside down, auto=false]{$\cong$}}},
         iso/.style={draw=none,every to/.append style={edge node={node [sloped, auto=false]{$\cong$}}}}}

\newcommand{\cpto}{
  \mathrel{\raisebox{0.5ex}{\kern3pt\ensuremath{\mathrel{\tikz{ \draw [-stealth,line width=0.4] (0.6ex,1ex) -- (0,1ex) -- (0,0.4ex) -- (2.2ex,0.4ex); }}}\kern3pt}}
}

\usetikzlibrary{decorations.markings}
\tikzset{
  double line with arrow/.style args={#1,#2}{
      decorate
    , decoration={markings
        , mark=at position 0 with {
            \coordinate (ta-base-1) at (0,1pt);
            \coordinate (ta-base-2) at (0,-1pt);
          }
        , mark=at position 1 with {
            \draw[#1] (ta-base-1) -- (0,1pt);
            \draw[#2] (ta-base-2) -- (0,-1pt);
          }
      }
  }
}

\let\oldemph\emph
\renewcommand{\emph}[1]{\oldemph{\index{#1}#1}}

\newcommand\sep{\mathrel{\ooalign{\kern1pt\rule[-1pt]{.41pt}{7pt}\kern1pt}}}

\renewcommand{\c}{\colon}

\newcommand{\klstar}{\star}  				%
\newcommand{\istar}{\dagger}  				%
\newcommand{\out}{\operatorname{\mathsf{out}}}
\newcommand{\inm}{\operatorname{\mathsf{in}}}
\newcommand{\tuo}{\operatorname{\out^{\text{\kern.5pt\rmfamily-}\kern-.5pt1}\kern-1pt}}

\newcommand{\nf}{\operatorname{\mathsf{nf}}}
\newcommand{\nm}{\operatorname{\mathsf{nm}}}

\renewcommand{\comp}{\mathbin{\operatorname*{\circ}}}					    %
\newcommand{\dcomp}[2]{{#2\comp }{#1}}

\newcommand{\ctx}{\vdash}
\newcommand{\ctxv}{\ctx_{\mathsf{v}}}
\newcommand{\ctxc}{\ctx_{\mathsf{c}}}

\newcommand{\gtag}{\bullet} %

\usepackage{etoolbox} %

\usepackage{ifdraft} 
\ifdraft{
  \usepackage{showlabels} 
  
  \usepackage[layout=footnote,draft]{fixme}

}{
 \usepackage[layout=footnote,final]{fixme}
}

\FXRegisterAuthor{sg}{asg}{SG}	%

\usepackage[english]{babel}

\addto\extrasenglish{%
}

\makeatletter
\newcommand{\superimpose}[2]{%
  {\ooalign{$#1\@firstoftwo#2$\cr\hfil$#1\@secondoftwo#2$\hfil\cr}}}
\makeatother

\newcommand{\lrule}[3]{\textbf{#1}~~\frac{#2}{#3}}

\newcommand{\IB}[2]{#1\mathbin{\hash}#2}

\newcommand{\G}{\Gamma}

\newcommand{\Sigv}{\Sigma_{\oname{v}}}
\newcommand{\Sigc}{\Sigma_{\oname{c}}}

\newcommand\gsep{
\operatorname{\,
\mathchoice{\raisebox{.3ex}{\scalebox{.8}[.8]{$\rangle\kern-4pt\rangle\kern-4pt\rangle\kern-4pt\rangle$}}}
           {\raisebox{.3ex}{\scalebox{.8}[.8]{$\rangle\kern-4pt\rangle\kern-4pt\rangle\kern-4pt\rangle$}}}
           {\raisebox{.15ex}{\scalebox{.7}[.5]{$\rangle\kern-4pt\rangle\kern-4pt\rangle$}}}
           {\raisebox{.15ex}{\scalebox{.7}[.5]{$\rangle\kern-4pt\rangle\kern-4pt\rangle$}}}
}\,}

\makeatletter
\newcommand{\IHom}{\@ifstar{\@ihomdn}{\@ihomup}}
\newcommand{\@ihomup}[1]{{#1}^{\kern.5pt\raisebox{.3pt}{\scalebox{.38}{$\blacksquare$}}}}
\newcommand{\@ihomdn}[1]{{#1}_{\kern.5pt\raisebox{.3pt}{\scalebox{.38}{$\blacksquare$}}}}

\newcommand{\GHom}{\@ifstar{\@ghomdn}{\@ghomup}}
\newcommand{\@ghomup}[1]{{#1}^{\kern-.2pt\gtag}}
\newcommand{\@ghomdn}[1]{{#1}_{\kern-.2pt\gtag}}

\let\dir@frac\frac
\newcommand{\inv@frac}[2]{\dir@frac{#2}{#1}}

\renewcommand{\frac}{\@ifstar{\inv@frac}{\dir@frac}}
\makeatother

\DeclareFontFamily{U}{mathc}{}
\DeclareFontShape{U}{mathc}{m}{it}{<->s*[1.03] mathc10}{}
\DeclareMathAlphabet{\morph}{U}{mathc}{m}{it}

\renewcommand{\id}{\operatorname{\morph{id}}}
\renewcommand{\inl}{\operatorname{\morph{inl}}}
\renewcommand{\inr}{\operatorname{\morph{inr}}}
\renewcommand{\fst}{\operatorname{\morph{fst}}}
\renewcommand{\snd}{\operatorname{\morph{snd}}}

\newcommand{\linj}{\oname{inl}}
\newcommand{\rinj}{\oname{inr}}

\newcommand{\proj}{\operatorname{\morph{proj}}}

\newcommand{\dist}{\operatorname{\morph{dist}}} 					%

\usepackage{accents}

\renewcommand{\wave}[1]{{\accentset{\scalebox{.95}[.9]{\texttildelow}}{#1}}}

\usepackage{bm}
\def\defbbname#1{\expandafter\def\csname BB#1\endcsname{{\bm{\mathsf{#1}}}}}
\def\defbbnames#1{\ifx#1\defbbnames\else\defbbname#1\expandafter\defbbnames\fi}
\defbbnames ABCDEFGHIJKLMNOPQRSTUVWXYZ\defbbnames

\newcommand{\nat}{\mathbb{N}}

\usepackage{todos}

\usepackage{savesym}

\savesymbol{degree}
\savesymbol{leftmoon}
\savesymbol{rightmoon}
\savesymbol{fullmoon}
\savesymbol{newmoon}
\savesymbol{diameter}
\savesymbol{emptyset}
\savesymbol{bigtimes}
\savesymbol{triangleright}
\savesymbol{trianglelefteq}

\usepackage[matha]{mathabx}
\restoresymbol{other}{emptyset}
\restoresymbol{other}{triangleright}
\restoresymbol{other}{trianglelefteq}

\usepackage{tikz}
\usepackage{tikz-cd}

\tikzset{
    commutative diagrams/.cd,
    arrow style=tikz,
    diagrams={>=stealth},
    row sep=large, 
    column sep = huge
}

\tikzstyle{shiftarr} = [
        rounded corners,%
        to path={--([#1]\tikztostart.center)
                     -- ([#1]\tikztotarget.center) \tikztonodes
                     -- (\tikztotarget)},
]

\usepackage{mathtools}

\newcommand{\mbind}[2]{\mfix{do}{#1;}{}{#2}}

\newcommand{\ibind}[2]{\mfix{iter}{#1;}{}{#2}}

\newcommand{\ret}{\oname{return}}

\newcommand\rsmraise[1]{%
  \ifx#1\displaystyle .8\else
    \ifx#1\textstyle .8\else
      \ifx#1\scriptstyle .6\else
        .45%
      \fi
    \fi
  \fi}

\renewcommand{\xto}[1]{\mathrel{\raisebox{-2pt}{$\;\xrightarrow{\scriptsize #1}$}}}

\DeclareRobustCommand{\pcase}[2]{\mfix{case\,}{#1}{\,of\,}{#2}}

\renewcommand{\dar}{\mathrel{\kern-.4pt\raisebox{1pt}{{\rotatebox[origin=c]{-90}{$\to$}}}\kern2pt}}
\renewcommand{\uar}{\kern-.2pt\mathrel{\raisebox{1pt}{{\rotatebox[origin=c]{90}{$\to$}}}}\kern-.2pt}

\renewcommand{\pm}{\mathbin{\mkern-02mu\hash\mkern-02mu}}
\newcommand{\ten}{\mathbin{\,\otimes\,}}

\usepackage{hypcap}
\allowdisplaybreaks

\begin{document}

\title[Representing Guardedness in Call-By-Value]
{Representing Guardedness in Call-by-Value\texorpdfstring{\\}{} and Guarded Parameterized Monads}
\thanks{Support by the Deutsche Forschungsgemeinschaft (DFG, German
  Research Foundation) is gratefully acknowledged -- project number 501369690}	%

\author[S.~Goncharov]{Sergey Goncharov\lmcsorcid{0000-0001-6924-8766}}

\address{University of Birmingham, Birmingham, UK}	%
\email{s.goncharov@bham.ac.uk}  %

\maketitle

\begin{abstract} 
Like the notion of computation via (strong) monads serves to classify various
flavours of impurity, including exceptions, non-determinism, probability, local
and global store, the notion of guardedness classifies well-behavedness of
cycles in various settings. In its most general form, the guardedness discipline
applies to general symmetric monoidal categories and further specializes to
Cartesian and co-Cartesian categories, where it governs guarded recursion and
guarded iteration, respectively. Here, even more specifically, we deal with the
semantics of call-by-value guarded iteration. It was shown by Levy, Power and
Thielecke that call-by-value languages can be generally interpreted in Freyd
categories, but in order to represent effectful function spaces, such a category
must canonically arise from a strong monad. We generalize this fact by showing
that representing \emph{guarded} effectful function spaces calls for certain
parameterized monads (in the sense of Uustalu). This provides a description of
guardedness as an intrinsic categorical property of programs, complementing the
existing description of guardedness as a predicate on a category.
\end{abstract}

\section{Introduction}
A traditional way to model call-by-value languages is based on a clear-cut separation between
computations and values. A computation can be \emph{suspended} and thus turned 
into a value, and a value can be \emph{executed}, and thus again be
turned into a computation. The paradigmatic example of these conversions is the 
application and abstraction mechanisms of the $\lambda$-calculus. From the categorical 
modelling perspective, this view naturally requires two categories, suitably connected with 
each other. As essentially suggested by Moggi~\cite{Moggi91a}, a minimal modelling
framework requires a Cartesian category (i.e.\ a category with finite products)
as a category of values and a Kleisli category of a strong monad over it, as a 
category of (side-effecting) computations (also called \emph{producers}~\cite{Levy04}). 
A generic \emph{computational metalanguage} thus arises
as an internal language of strong monads. Levy, Power and Thielecke~\cite{LevyPowerEtAl02}
designed a refinement of Moggi's computational metalanguage, called \emph{fine-grain 
call-by-value (FGCBV)}, whose models are not necessarily strong monads, but are more 
general \emph{Freyd categories}. They have shown that a strong monad in fact 
always emerges from a Freyd category if certain function spaces
(needed to interpret higher-order functions) are~\emph{representable} as objects of 
the value category -- thus strong monads arise from first principles.

Here, we analyse an extension of the FGCBV paradigm with a notion of 
guardedness, which is a certain predicate on computations, certifying their 
well-behavedness, in particular that they can be iterated~\cite{GoncharovSchroderEtAl17,LevyGoncharov19}. 
A typical example is guardedness in process algebra, where guardedness is often used 
to ensure that recursive systems of process definitions have unique solutions~\cite{Milner89}. 

FGCBV does not directly deal with fixpoints, since these are 
usually considered to be features orthogonal to computational effects and evaluation 
strategies. Analogously, even though the notion of guardedness is motivated
by fixpoints, here we do not consider (guarded) fixpoints as a core language feature. 
In fact, in practically relevant cases guardedness is meaningful on its own as 
a suitable notion of \emph{productivity} of computation, and need not be justified via
fixpoints, which may or may not exist. In FGCBV, one typically regards general recursion to be supported by the category 
of values, and once the latter indeed does so (e.g.\ by being a suitable category of 
complete partial orders), it is obvious to add a corresponding fixpoint construct to the language. 

Let us nevertheless outline the connection between guardedness and recursion in some 
more detail. General recursion entails partiality for programs, meaning that
even if we abstract from it, the corresponding effect of \emph{partiality}
must be part of the computational effect abstraction (see e.g.~\cite{Fiore04}). 
Recursion and computational effects are thus intimately connected. This connection
persists under the restriction from general recursion to iteration, which 
is subject to a much broader range of models, and triggers the partiality effect
just as well.
Arguably, the largest class of monads, supporting iteration, are
\emph{Elgot monads}~\cite{AdamekMiliusEtAl10,GoncharovRauchEtAl15}. These are
monads~$T$, equipped with \emph{Elgot iteration}:
\begin{align}\label{eq:total}
\lrule{}{f\c X\to T(Y+X)}{f^\istar\c X\to TY}
\end{align}
and subject to established equational laws~\cite{BloomEsik93,SimpsonPlotkin00}.
Intuitively, $f^\istar$ is obtained from~$f$ by iterating away the right summand 
in the output type $Y+X$. For example, the maybe-monad $(\argument)+1$ is an Elgot 
monad over the category of \emph{classical} sets, which yields a model for a 
while-language with non-termination as the only computational effect. 
Now, \emph{guarded Elgot monads}~\cite{LevyGoncharov19} refine Elgot monads in that
the operator~\eqref{eq:total} needs only to be defined w.r.t.\ a custom class 
of \emph{guarded morphisms}, governed by simple laws. Proper partiality of the guardedness 
predicate is relevant for various reasons, such as the following.

\begin{figure*}[t]
\begin{equation*}
~~\begin{tikzcd}[column sep = -2ex,row sep = 2ex]
 &[10ex]  & \parbox{3cm}{\scriptsize guarded\\ parameterized monads\footnotemark{}} \arrow[ddd,<-] \arrow[rrr] &  &  &[10ex] \parbox{3.6cm}{\scriptsize strong guarded\\parameterized monads\textsuperscript{\ref{foot}}} \\
 &  &  &  &  &  \\
\parbox{2.5cm}{\scriptsize id.-on-obj.\ guarded\\ functors} \arrow[rruu] \arrow[rrr,crossing over] &  &  & \parbox{2cm}{\scriptsize guarded Freyd\\ categories}  \arrow[rruu] &  &  \\
 &  & \text{\scriptsize monads} \arrow[rrr] &  &  & \text{\scriptsize strong monads} \arrow[uuu] \\
 &  &  &  &  &  \\[3ex]
\text{\scriptsize id.-on-obj. functors}
  \arrow{rruu}[sloped,above]{\text{\tiny representability}}
  \arrow{uuu}[sloped,above]{\text{\tiny guardedness}}
  \arrow{rrr}[sloped,above]{\text{\tiny strength}}
  &  &  & \text{\scriptsize Freyd categories} \arrow[rruu] \arrow[uuu,crossing over] &  &
\end{tikzcd}
\end{equation*}
  \caption{Three dimensions within call-by-value.}
  \label{fig:cube}
\end{figure*}
\begin{itemize}%
  \item Guarded fixpoints often \emph{uniquely} satisfy the corresponding fixpoint equation~\cite{Uustalu03,Milius05,GoncharovSchroderEtAl19},
  which greatly facilitates reasoning; this is extensively used in bisimulation-based
  process algebra~\cite{Milner89,Fokkink13}.
  \item In a type-theoretic and constructive setting, guarded iteration can often be 
  defined natively and more generally, e.g.\ the ``simplest'' guarded Elgot monad is
  Capretta's \emph{delay monad} (initially called ``partiality monad'')~\cite{Capretta05}, rendered by final coalgebras $D=\nu\gamma.\,(\argument+\gamma)$, 
  which yields an intensional counterpart of the maybe-monad; guardedness then 
  means \emph{productivity}, i.e., that the computation signals that it evolves if it does.
  Contrastingly, the ``simplest'' Elgot monad is 
  much harder to construct and arguably requires additional principles to
  be available in the underlying metatheory~\cite{ChapmanUustaluEtAl15,AltenkirchDanielssonEtAl17,EscardoKnapp17,Goncharov21}.
  \item Guardedness is a compositional type discipline, and hence it potentially 
  helps to encapsulate additional information about the productivity of programs in types, 
  like monads encapsulate the information about potential side-effects.
\end{itemize}
\footnotetext{More precisely, representability yields parameterized guarded monads,
subject to an additional monicity condition. This is treated in detail in~\autoref{sec:gpm}.\label{foot}}
By allowing the iteration operator to be properly partial, we can accommodate 
a range of new examples of iterative behaviour. A notion of guardedness thus 
often plays an auxiliary role of determining, in a compositional way,
which morphisms can be iterated.
   
As indicated above, strong monads can be regarded as structures, in a canonical 
way arising from FGCBV by adding the requirement of representability of certain
function spaces in the category of values. This is behind the mechanism of representing
computational effects via monads in type systems (e.g.\ in F$\omega$, by quantification
over higher kinds) and hence in programming languages (e.g.\ in Haskell). Our goal is to provide an analogous
mechanism for guardedness and for its combinations with computational effects 
and strength. That is, (strong) monads are an answer to the question: 

\begin{quote}{}
\emph{What is the 
categorical/type-theoretic structure that faithfully represents computational
effects within a higher-order universe?}
\end{quote}
Here, we are answering the question: 
\begin{quote}{}
\emph{What is the 
categorical/type-theoretic structure that faithfully represents guarded 
computational effects within a higher-order universe?} 
\end{quote}
In other words, we seek to formulate guardedness as an intrinsic structural property of morphisms, rather than as additional data that (anonymously) identifies guarded morphisms among others. In doing so, we are inspired by the view of monads as structures for representing effects, as summarized above.
In fact, we show that strength, representability and guardedness can be naturally 
arranged  within FGCBV as three orthogonal dimensions, as shown in \autoref{fig:cube} (the arrows point from more general concepts
to more specific ones). The bottom face of the cube features the above-mentioned 
connection between Freyd categories and strong monads, and a corresponding connection
between identity-on-object functors and (not necessarily strong) monads. We contribute
with the top face, where guardedness is combined with other dimensions. The key
point is the combination of guardedness with representability,
which produces a certain class of \emph{parameterized monads}~\cite{Uustalu03}
that we dub \emph{guarded parameterized~monads}.

\smallskip
\textbf{Related work~} We benefit from the analysis of Power and 
Robinson~\cite{PowerRobinson97}, who introduced \emph{premonoidal categories} 
as an abstraction of Kleisli categories. Freyd categories were subsequently 
defined by Power and Thielecke~\cite{PowerThielecke99} as premonoidal categories 
with additional structure and also connected to strong monads. Levy~\cite{Levy04} 
came up with an equivalent definition, which we use throughout. In the previous 
characterization~\cite{PowerThielecke99,LevyPowerEtAl02}, strong monads were shown to arise jointly 
with Kleisli exponentials from \emph{closed Freyd categories}. We refine this 
characterization (\autoref{cor:str-freyd}) by showing that strong monads in fact 
arise independently of exponentials (\autoref{prop:freyd}). \emph{Distributive
Freyd categories} were defined by Staton~\cite{Staton14} -- here we use them
to extend the FGCBV language with coproducts and, subsequently, with guardedness predicates.
Previous approaches to identifying structures for 
ensuring guardedness on monads involved \emph{monad modules}~\cite{PirogGibbons14,JirAdamekMiliusEtAl02}
-- we make do with guarded parameterized monads instead, which combine monads with
modules over them and arise universally. \sgnote{Mention Paul-Andre's work.}

\textbf{Plan of the paper~} 
After short technical preliminaries, we start off by introducing a restricted version of 
FGCBV in~\autoref{sec:simp} and extensively discuss motivating examples, which 
(with a little effort) can already be encoded despite restrictions. We establish 
a very simple form of the representability scenario, producing monads, and meant 
to serve as a model for subsequent sections. In~\autoref{sec:mon} we deal with 
full FGCBV, Freyd categories, modelling them and strong monads, representing Freyd 
categories. The guardedness dimension is introduced in~\autoref{sec:gfreyd} where 
we define \emph{guarded Freyd categories}, and in~\autoref{sec:repr} we analyse
the representability issue for them. Finally, in~\autoref{sec:gpm} we introduce 
an equational axiomatization of a categorical structure for representing guardedness, 
called \emph{guarded parameterized monads}. As a crucial technical step, we establish 
a coherence property in the style of Mac Lane's coherence theorem for 
monoidal categories~\cite{Mac-Lane71}.

The present paper is an extended version of the conference paper~\cite{Goncharov23}.
We added the proofs and more details to the examples and the general discussion. 
The original definition~\cite{Goncharov23} of the guarded parameterized monad
was missing two coherence conditions, which are now added (\autoref{def:pm}).

\section{Preliminaries}\label{sec:prelim}
We assume familiarity with the basics of category 
theory~\cite{Mac-Lane71,Awodey10}. For a category~$\BV$, $|\BV|$ will denote the 
class of objects, and $\BV(X,Y)$ will denote morphisms from~$X$ to $Y$. We tend to
omit indices at natural transformations for readability.
A category with finite (co\dash)products is called (co\dash)Cartesian. 
In a co-Cartesian category with selected
coproducts, we write 
$\bang\c\iobj\to A$ for the initial morphism, and $\inl\c A\to A+B$ and $\inr\c B\to A+B$
for the left and right coproduct injections, respectively.
A \emph{distributive category}~\cite{Cockett93} is a Cartesian and co-Cartesian category, in which
the natural transformation
\begin{align*}
X\times Y + X\times Z \xrightarrow{~~[\id\times\inl,\;\id\times\inr]~~} X\times (Y+Z)
\end{align*}
is an isomorphism, whose inverse we denote $\dist_{X,Y,Z}$ (a 
co-Cartesian and Cartesian closed category is always distributive). Let $\Delta=\brks{\id,\id}\c X\to X\times X$ 
and $\nabla=[\id,\id]\c X+X\to X$.
A monad~$\BBT$ on~$\BV$ is determined by a \emph{Kleisli triple}~$(T,\eta, (-)^\star)$,
consisting of a map $T\c{|\BV|\to|\BV|}$, a family of morphisms 
$(\eta_X\c X \to TX)_{X\in|\BV|}$ and \emph{Kleisli lifting} sending each~$f\c X \to T Y$ to~$f^\star\c TX \to TY$ and
obeying \emph{monad laws}: %
\begin{align*}
\eta^{\klstar}=\id, && 
f^{\klstar}\comp\eta=f,  && 
(f^{\klstar}\comp g)^{\klstar}=f^{\klstar}\comp g^{\klstar}.
\end{align*}
It follows that~$T$ extends to a functor, $\eta$ extends to a natural transformation -- \emph{unit},
$\mu = \id^\klstar\c TTX\to TX$ extends 
to a natural transformation -- \emph{multiplication}, and that $(T,\eta,\mu)$ 
is a monad in the standard sense~\cite{Mac-Lane71}. We will generally use blackboard capitals (such as~$\BBT$) to refer to monads 
and the corresponding Roman letters (such as~$T$) to refer to their functor parts.
Morphisms of the form~$f\c X\to TY$ are called \emph{Kleisli morphisms} and form the \emph{Kleisli
category}~$\BV_{\BBT}$ of $\BBT$ under \emph{Kleisli composition} $f,g\mto f^\klstar\comp g$ 
with identity~$\eta$. 

An endofunctor $F$ is \emph{strong} if it is equipped with a natural transformation \emph{strength}
$\tau\c X\times FY\to F(X\times Y)$, such that the diagrams
\begin{equation*}
\begin{tikzcd}[column sep = 5ex,row sep = 4ex]
1\times FX
  \dar["\tau"']
  \rar["\snd"]
&
FX
\\
F(1\times X)
  \urar["F\snd"'] &
\end{tikzcd}
\qquad
\begin{tikzcd}[column sep = 2ex,row sep = 4ex]
(X\times Y)\times FZ
  \dar["\iso"]
  \ar[rr,"\tau"]
&[3ex] &
F((X\times Y)\times Z)
  \dar["\iso"]
\\
X\times (Y\times FY)
  \rar["\id\times\tau"]
&
X\times F(Y\times Z)
  \rar["\tau"]
&
F(X\times (Y\times Z))
\end{tikzcd}
\end{equation*}
commute. A natural transformation between two strong functors is strong if it 
preserves strength in the obvious sense, and a monad $\BBT$ is strong if $T$ is 
strong with some strength $\tau\c X\times TY\to T(X\times Y)$ and $\eta$ and $\mu$ are strong with $\id$ being the strength of $\Id$ 
and $T\tau\comp\tau\c X\times TTY\to TT(X\times Y)$ being the strength of $TT$.

\section{Simple FGCBV with Coproducts}\label{sec:simp} %
We start off with a restricted -- single-variable -- fragment of FGCBV,
but extended with coproduct types. Since we will not deal with operational semantics,
we simplify the language slightly (e.g.\ we do not include let-expressions for values).
We also stick 
to a Haskell-style syntax with do-notation and case expressions.
We fix a collection of sorts~$S_1,S_2,\dots$, a signature~$\Sigv$ 
of pure programs $f\c A\to B$, and a signature~$\Sigc$ of effectful programs~$f\c A\to B$
(also called \emph{generic effects}~\cite{PlotkinPower01}) where $A$ and $B$ are types, generated with the grammar
\begin{align}\label{eq:types}
A,B &\Coloneqq S_1,S_2,\ldots\mid 0\mid A + B.
\end{align}
We then define terms in context of the form~$x\c A\ctxv v\c B$ and~$x\c A\ctxc p\c B$
for value terms and computation terms inductively by the rules given in \autoref{fig:simple}.
(where we chose to stick to the syntax of the familiar Haskell's do-notation):
\begin{figure*}[t]
 \begin{equation*}%
\begin{gathered}
\frac{}{~x\c A\ctxv x\c A~}\qquad
\frac{f\c A\to B\in\Sigv\quad\G\ctxv v\c A}{~\G\ctxv f(v)\c B~}\qquad
\frac{f\c A\to B\in\Sigc\quad\G\ctxv v\c A}{~\G\ctxc f(v)\c B~}\\[2ex]
\frac{~\G\ctxv v\c A~}{\G\ctxc \ret v\c A}\qquad
\frac{~\G\ctxc p\c A\quad x\c A\ctxc q\c B}{\G\ctxc \mbind{x\gets p}{q}\c B}\qquad
\frac{\G\ctxv v\c 0}{\G\ctxc\oname{init} v\c A}
\\[3ex]
\frac{~\G\ctxv v\c A~}{\G\ctxv \linj v\c A+B}\qquad\!
\frac{~\G\ctxv v\c B~}{\G\ctxv \rinj v\c A+B}\qquad\!
 \frac{%
   \G\ctxv v\c A+B\quad x\c A\ctxc p\c C\quad y\c B\ctxc q\c C
  }{%
   \G\ctxc\case{v}{\linj x\mto p}{\rinj y\mto q}\c C
  }
\end{gathered}
\end{equation*}
 \caption{Simple FGCBV with coproducts.}
  \label{fig:simple}
\end{figure*}
This language is essentially a refinement of Moggi's \emph{simple (!)
computational metalanguage}~\cite{Moggi91a}, which has only one-variable contexts (i.e.\ $\G$ is of the form $x\c A$
throughout), rather than fully fledged multi-variable contexts. In terms of monads, the present language corresponds to not necessarily strong ones.
In terms of monads, the present language
corresponds to not necessarily strong ones. Such monads are not
very useful in traditional programming languages semantics; however we dwell on this 
case for several reasons. We aim to explore the interaction between guardedness 
and monads from a foundational perspective, while remaining as general as possible to cover
cases where strength does not exist or is not relevant. %
We would also like to identify the basic representation scenario, 
to be extended later to more sophisticated cases.

An obvious extension of the presented language would be the iteration
operator: 
\begin{align}\label{eq:iter}
\frac{\G\ctxc p\c A\qquad x\c A\ctxc q\c B+A}{\G\ctxc\ibind{x\gets p}{q}\c B}
\end{align}
meant to satisfy the fixpoint equality 
\begin{align*}
	\ibind{x\gets p}{q} = \ibind{x\gets (\mbind{x\gets p}{q})}{q}
\end{align*}
Presently, we focus on representing guardedness 
as such and do not deal with~\eqref{eq:iter}

We present three examples that can be interpreted w.r.t.\ the 
single-variable fragment to demonstrate the unifying power of FGCBV
and illustrate various flavours of guardedness.
\begin{example}[Basic Process Algebra~\cite{BergstraPonseEtAl01}]\label{exa:bpl}
\emph{Basic process algebra (BPA)} over a set of actions~$\CA$ is defined by the 
grammar:
\begin{align*}
P,Q  \Coloneqq (a\in\CA)\mid P+Q\mid P\cdot Q. 
\end{align*}
One typically considers BPA-terms over free variables (seen as process names) to 
solve systems of recursive process equations w.r.t.\ these variables. E.g.\ we 
can specify a 2-bit FIFO buffer as a solution to
\begin{equation}\label{eq:bpa}
\begin{aligned}
B_0 			=&\; \oname{in}_0\cdot B_1^0 + \oname{in}_1\cdot B_1^1\\
B_1^i 		=&\; \oname{in}_0\cdot B_2^{0,i} + \oname{in}_1\cdot B_2^{1,i} + \oname{out}_i\cdot B_0&&\qquad (i\in\{0,1\})\\
B_2^{i,j} =&\; \oname{out}_j\cdot B_1^i&&\qquad (i,j\in\{0,1\})
\end{aligned}
\end{equation}
with $\CA = \{\oname{in}_0,\oname{in}_1,\oname{out}_0,\oname{out}_1\}$.
We view $B_0$ as an empty FIFO, $B_1^i$ as a FIFO carrying
only~$i$ and $B_2^{i,j}$ as a FIFO carrying~$i$ and $j$. For example, the trace 
\begin{align*}
B_0\trans{\oname{in}_0} B_1^0 \trans{\oname{in}_1} B_2^{1,0} \trans{\oname{out}_0} B_1^1\trans{\oname{out}_1} B_0
\end{align*}
is valid and represents the following course of action: push~$0$, push $1$, pop $1$
and then pop $0$. 
We can model such systems of equations in FGCBV as follows. Let us fix a single sort~$1$
and identify an $n$-fold sum $(\ldots (1+\ldots)\ldots)+1$ with the natural number~$n$. 
The injections $\inj_i\c 1\to n$ are defined inductively in the obvious way. 
Let $\Sigv=\emptyset$ and $\Sigc = \{a\c 1\to 1\mid a\in\CA\}\cup\{\oname{toss}\c {1\to 2}\}$.
A BPA-term over process names $\{N_1,\ldots,N_n\}$ can be translated to 
FGCBV recursively, with the following rules where $\rightsquigarrow$ reads as ``translates'':

\begin{equation*}
\begin{gathered}
\frac{}{N_i\rightsquigarrow x\c 1\ctxc\ret(\rinj (\inj_i x))\c  1+n}\qquad
\frac{}{a\rightsquigarrow x\c 1\ctxc \mbind{x\gets a(x)}{\ret(\linj x)}\c  1+n}~~\\[2ex]
\frac{P\rightsquigarrow x\c 1\ctxc p\c 1+n\qquad Q\rightsquigarrow x\c 1\ctxc q\c 1+n}{P+Q\rightsquigarrow x\c 1\ctxc\mbind{x\gets \oname{toss}(x)}{\case{x}{\linj x\mto p}{\rinj x\mto q}}\c 1+n}\\[2ex]
\frac{P\rightsquigarrow x\c 1\ctxc p\c 1+n\qquad Q\rightsquigarrow x\c 1\ctxc q\c 1+n}{P\cdot Q\rightsquigarrow x\c 1\ctxc\mbind{x\gets p}{\case{x}{\linj x\mto q(x)}{\rinj x\mto\ret(\rinj x)}}\c 1+n}
\end{gathered}
\end{equation*}
Intuitively, the terms $x\c 1\ctxc p\c 1+n$ represent processes with 
$1+n$ exit points: every process name $N_i$ identifies an exit $i$, in 
addition to the global anonymous exit. The latter is associated with actions, which 
are not postcomposed with any other commands. %
The generic effect $\oname{toss}$ induces binary nondeterminism as a coin-tossing act. 
For example, the result of translating the right-hand sides of~\eqref{eq:bpa} 
(after minor simplifications) is
\begin{align*}
&\;\mbind{x\gets \oname{toss}(x)}{\case{x}{\\
	&\hspace{1.75cm}\linj x\mto \mbind{x\gets \oname{in}_0(x)}{\ret(\rinj (\inj_1^0 x))}}{\\
	&\hspace{1.75cm}\rinj x\mto \mbind{x\gets \oname{in}_1(x)}{\ret(\rinj (\inj_1^1 x))}}},\\[1.5ex]
&\;\mbind{x\gets \oname{toss}(x)}{\case{x}{\\
	&\hspace{1.75cm}\linj x\mto \mbind{x\gets \oname{in}_0(x)}{\ret(\rinj (\inj_2^{0,i} x))}}{\\
	&\hspace{1.75cm}\rinj x\mto \mbind{x\gets \oname{toss}(x)}{\case{x}{\\
	&\hspace{2.75cm}\linj x\mto \mbind{x\gets \oname{in}_1(x)}{\ret(\rinj (\inj_2^{1,i} x))}}{\\
	&\hspace{2.75cm}\rinj x\mto \mbind{x\gets \oname{out}_i(x)}{\ret(\rinj (\inj_0 x))}}}}}&&\qquad (i\in\{0,1\})\\[1.5ex]
&\; \mbind{x\gets \oname{out}_j(x)}{\ret(\rinj (\inj_1^i x))}&&\qquad (i,j\in\{0,1\})
\end{align*}
where $n=1 + 2 + 4 = 7$, $\inj_0\c 1\to 7$, the $\inj_1^i\c 1\to 7$ and the $\inj_1^{i,j}\c 1\to 7$
are the injections, selecting the indices that address $B_0$, $B_1^{i}$ 
and $B_2^{i,j}$ correspondingly. 
Every list of terms in the context $(x\c 1\ctxc p_0\c m),\ldots,(x\c 1\ctxc p_{n-1}\c m)$ 
can be converted to a single term $x\c n\ctxc\hat p_n\c m$ recursively as follows: 
\begin{align*}
\hat p_0 = \oname{init} x,\qquad \hat p_{n+1} = \case{x}{\linj x\mto \hat p_{n}}{\rinj x\mto p_{n+1}}.
\end{align*}
Every system of $n$ equations over $m+n$ process names in BPA is thus represented 
by a term $x\c n\ctxc p\c (1+m)+n$ in simple FGCBV. Now, an iteration operator~\eqref{eq:iter} 
applied to the latter term ``solves'' the corresponding system of equations 
w.r.t.\ to $n$ names, and keeping the remaining $m$ names free, resulting in a 
term of the form $x\c n\ctxc\ibind{x\gets\ret x}{p}\c 1+m$. In our example~\eqref{eq:bpa},
$n=m=7$.

Guarded systems are those in which recursive calls are preceded by actions;~\eqref{eq:bpa}
is an example. Such systems have a unique solution (under bisimilarity)~\cite{BaetenWeijland91,Fokkink13}. The simplest unguarded example $P=P$ has arbitrary solutions and translates to $x\c 1\ctxc\ibind{x\gets\ret x}{\ret(\rinj x)}\c 1$.
\end{example}

\begin{example}[Imperative Traces]\label{exa:itrace}
We adapt the semantic framework of Nakata and Uustalu~\cite{NakataUustalu15} for 
imperative coinductive traces to our setting. 
Let us fix a set $P$ of 
\emph{predicates}, a set $T$ of \emph{state transformers}, and let the corresponding
pure and effectful signatures be 
$\Sigv = \{p\c S\to S + S\mid p\in P\}\cup\{t\c S\to S\mid t\in T\}$ 
and $\Sigc = \{\oname{put}\c S\to 1\comma\oname{get}\c 1\to S\}$ over the set of 
sorts $\{S,1\}$. The intended interpretation of this data is as follows:
\begin{itemize}
  \item $S$ is a set of memory states, e.g.\ the set of finitely supported partial 
  functions~$\nat\ito 2$;
  \item $T$ are state transformers, e.g.\ functions,
updating precisely one specified memory bit; 
  \item $p\in P$ encode predicates: $p(s) = \inl(s)$ if the predicate 
  is satisfied and $p(s)=\inr(s)$ otherwise, e.g.~$p$ can capture functions that give 
a Boolean answer to the questions ``is the specified bit~$0$?'' and ``is the specified bit~$1$?''.  
\end{itemize}
For example, the following program negates the $i$-th memory bit (if it is present)
\begin{align*}
  x\c 1\ctxc\mbind{s\gets\oname{get}(x)}{\case{(s[i]=0)}{\linj s\mto\oname{put}(s[i\coloneqq1])}{\rinj s\mto\oname{put}(s[i\coloneqq0])}}\c 1,
\end{align*}
where $(\argument[i]=0)$, $(\argument[i\coloneqq0])$ and
$(\argument[i\coloneqq1])$ are the obvious predicate and state transformers.
Nakata and Uustalu~\cite{NakataUustalu15} argued in favour of (infinite) traces 
as a particularly suitable semantics for reasoning about imperative programs. 
This means that store updates must contribute to the semantics, 
which can be ensured by a judicious choice of syntax, e.g., by using 
$\oname{skip}=\mbind{s\gets\oname{get}(x)}{\oname{put}(s)}$, but not $\ret$. 
In FGCBV, however, iterating $x\c 1\ctxc\ret(\rinj x)\c 1$ would not yield any trace.
By restricting to guarded iteration, with guardedness meaning writing to the store, we can indeed prevent such programs from iterating by defining guardedness so that at least one $\oname{put}$ is executed before the body of the loop is repeated. 

\end{example}
\begin{example}[Hybrid Programs]\label{exa:hybrid}
Hybrid programs combine discrete and continuous capabilities and can thus be 
used to describe the behaviours of cyber-physical systems. %
For simplicity, we consider time delays as the only hybrid facility -- more sophisticated
scenarios are treated elsewhere~\cite{GoncharovNevesEtAl20}. 
Let~$\realp$ be 
the sort of non-negative real numbers and let $\Sigv$ contain all unary operations
on non-negative reals and additionally $\oname{is}_0\c\realp\to\realp+\realp$,
which sends $n=0$ to $\inl(n)$ and $n>0$ to $\inr(n)$. Let $\Sigc=\{\oname{wait}\c{\realp\to\realp}\}$. With $\oname{wait}(r)$
we can introduce a time delay of length~$r$ and return $r$. With iteration we can write programs like
\begin{align*}
x\c \realp\ctxc\ibind{x\gets\ret x}{\quad&\\\quad\case{\oname{is}_0(x)}{\linj x&\mto\ret(\linj x)}{\\*\rinj x&\mto\mbind{x\gets\oname{wait}(x)}{\ret (\rinj f(x))}}}\c\realp,
\end{align*}
which terminate successfully in finite time ($f(x) = x\dotminus 1$\footnote{$\dotminus$
refers to truncated subtraction: $x\dotminus y = x - y$ if $x\geq y$, and $x\dotminus y=0$ otherwise.}), run infinitely
(${f(x) = 1}$), or exhibit \emph{Zeno behaviour} ($f(x) = x / 2$), i.e.\ consume 
finite time, but never terminate. In all these examples, every iteration consumes 
non-zero time. This is also often considered a well-behavedness condition, 
which we can naturally interpret as guardedness. 
\end{example}
To interpret the language from~\autoref{fig:simple}, let us fix two 
co-Cartesian categories~$\BV$ and~$\BC$, and an identity-on-objects functor~$J\c{\BV\to\BC}$ 
(hence~$|\BV|=|\BC|$) that strictly preserves coproducts. 
A semantics of~$(\Sigv,\Sigc)$ over $J$ assigns 
\begin{itemize}
  \item an object~$\sem{A}\in |\BV|$ to each sort~$A$;
  \item a morphism~$\sem{f}\in\BV(\sem{A},\sem{B})$ to each~$f\c A\to B\in\Sigv$;
  \item a morphism~$\sem{f}\in\BC(\sem{A},\sem{B})$ to each~$f\c A\to B\in\Sigc$,
\end{itemize}
which extends to types as follows: $\sem{0} = 0$, $\sem{A+B} = \sem{A}+\sem{B}$.  
The semantics of terms are given in~\autoref{fig:simple-sem}.
\begin{figure*}[t]
\begin{gather*}
\frac{}{\sem{x\c A\ctxv x\c A} = \id} 
\qquad
\frac{h=\sem{\G\ctxv v\c A}}{\sem{\G\ctxv f(v)\c B} = \sem{f}\comp h}
\qquad  
\frac{h=\sem{\G\ctxv v\c A}}{\sem{\G\ctxc f(v)\c B} = \sem{f}\comp Jh}
\\[2ex] 
\frac{h=\sem{\G\ctxv v\c A}}{\sem{\G\ctxc \ret v\c A} = Jh}
\qquad
\frac{h_1=\sem{x\c A\ctxc q\c B}\qquad h_2=\sem{\G\ctxc p\c A}}{\sem{\G\ctxc \mbind{x\gets p}{q}\c B} = h_1\comp h_2}
\\[2ex]
\frac{}{\sem{\G\ctxc\oname{init} v\c A} = \bang}
\qquad
\frac{h=\sem{\G\ctxv v\c A}}{\sem{\G\ctxv\linj v\c A+B} = \inl\comp h}
\qquad
\frac{h=\sem{\G\ctxv v\c B}}{\sem{\G\ctxv\rinj v\c A+B} = \inr\comp h}
\\[2ex]
\frac{h=\sem{\G\ctxv v\c A+B}\qquad h_1=\sem{x\c A\ctxc p\c C}\qquad h_2=\sem{y\c B\ctxc q\c C}}{\sem{\G\ctxc\case{v}{\linj x\mto p}{\rinj y\mto q}\c C} = [h_1,h_2]\comp Jh}
\end{gather*}
 \caption{Denotational semantics of simple FGCBV with coproducts.}
  \label{fig:simple-sem}
\end{figure*}
As observed by Power and Robinson~\cite{PowerRobinson97} (cf.~\cite[0.1]{Schumacher69}),
monads arise from the requirement that~$J$ is a left adjoint, thus simple FGCBV can be 
interpreted w.r.t.\ a monad on~$\BV$. A direct simple proof is given below for the 
sake of completeness. 
\begin{proposition}\label{prop:kleisli}
Let~$J\c\BV\to\BC$ be an identity-on-objects functor. Then~$J$ is a left adjoint 
iff\/~$\BC$ is isomorphic to a Kleisli category of some monad\/~$\BBT$ on~$\BV$ and 
$J = H\comp J_\BBT$ where~$H\c\BV_\BBT\iso\BC$ is the relevant isomorphism,
which is necessarily identity-on-objects, and $J_\BBT\c\BV\to\BV_\BBT$ is 
the canonical left adjoint sending every $f\in\BV(X,Y)$ to $\eta\comp f\in\BV(X,TY)$.

Moreover, in this situation, finite coproducts in $\BC$ are inherited from $\BV$, i.e.~$J\bang$
is the initial morphism in $\BC$ and the triples $(X+Y,J\inl,J\inr)$ are binary coproducts in~$\BC$.
\end{proposition}
\begin{proof}
Suppose that~$J\dashv U$ and consider the diagram
\begin{equation*}
\begin{tikzcd}[column sep=large, row sep=normal]
\BV_{\BBT}\rar["K_{\BBT}"] & \BC\\ %
& \BV\ular[bend left=0,"J_{\BBT}"]\uar["J"'] %
\end{tikzcd}
\end{equation*}
where $K_{\BBT}$ is the comparison functor from the Kleisli category of $\BBT$ to $\BC$.
Note that $K_\BBT$ is generally full and faithful,
because $\BV_\BBT(X,Y)=\BV(X,UJY)\iso \BC(JX,JY) = \BC(K_\BBT X,K_\BBT Y)$.
Moreover, $K_\BBT$ is identity-on-objects, for so is $J$ by assumption. 
Thus, $K_\BBT$ is an isomorphism and $J f = (K_\BBT\comp J_\BBT)(f) = K_\BBT (\eta\comp f)$ for any~$f\in\BV(X,Y)$.

Now, suppose that for a suitable monad~$\BBT$,~$H\c\BV_\BBT\iso\BC$ and 
$J f = H(\eta\comp f)$ for any~$f\in\BV(X,Y)$. Let~$U\vdash J$ be the adjunction
between~$\BV$ and~$\BV_\BBT$, and show that~$UH^\mone\vdash J$. Note that 
$H^\mone\vdash H$, and hence, by composing adjunctions~$UH^\mone\vdash HJ$. For 
every~$f\in\BV(X,Y)$,~$HJ f = H(\eta\comp f) = Jf$, i.e.\ indeed,~$UH^\mone\vdash J$.

That finite coproducts in $\BC$ are inherited from $\BV$ is easy to see.
\end{proof}
\begin{example}[Monads]\label{exa:monad}
Let us recall relevant monads on $\BV=\Set$ for further reference.
\begin{enumerate}[wide]
  \item\label{it:mon1} $TX=\nu\gamma.\,\FSet((X+1)+A\times\gamma)$ where $\FSet$ is the finite powerset
  functor and $\nu\gamma.\, F\gamma$ denotes a final $F$-coalgebra. This monad provides a standard strong bisimulation 
  semantics for BPA (\autoref{exa:bpl}). The denotations in $TX$ are finitely 
  branching trees with edges labelled by actions and with terminal nodes labelled
  in $X$ (free variables) or in $1$ (successful termination). This monad is an instance
  of the \emph{coinductive resumption monad}~\cite{PirogGibbons14}, and the inhabitants of $TX$
  are often called \emph{synchronization trees} (e.g.~\cite{AcetoCarayolEtAl12}).
  \item\label{it:mon2} $TX=\PSet(A^\star\times (X+1) + A^\star)$ is the monad of finite traces
  (terminating successfully $A^\star\times (X+1)$ and divergent $A^\star$), which can again 
  be used as a semantics of \autoref{exa:bpl}. 
  \item $TX=\PSet(A^\star\times (X+1) + (A^\star+A^\omega))$ is a refinement of~\eqref{it:mon2},
  collecting not only finite, but also infinite traces. If we extend BPA with countable 
  non-determinism, we obtain a semantics properly between strong bisimilarity and
  finite trace equivalence. For example, the equation $P=a\cdot P$ 
  produces the infinite trace $a^\omega$ and $P'=\sum_{i\in\nat} P_i$ with 
  $P_0 = a$ and $P_{i+1} = a\cdot P_i$ do not. Therefore, $P$ is not infinite trace 
  equivalent to $P'$, while $P$ and~$P'$ are finite trace equivalent. 
  \item $TX=(\nu\gamma.\,X\times S+\gamma\times S)^S$ can be used for \autoref{exa:itrace}.
  In $\Set$, $TX\iso (X\times S^\mplus + S^\omega)^S$, i.e.\ an element 
  $TX$ is isomorphic to a function that takes an initial state in $S$ and returns 
  either a finite trace in $X\times S^\mplus$ or an infinite trace in $S^\omega$.
  We can use \autoref{prop:kleisli} to argue that~$T$ indeed extends to a 
  monad. Indeed, let $\BC$ be the category with $\BC(X,Y)=\Set(X\times S, \nu\gamma.\,Y\times S+\gamma\times S)$,
  which is a full subcategory of the Kleisli category of the coinductive resumption 
  monad $\nu\gamma.\,(\argument+\gamma\times S)$. Now, by definition, the obvious identity-on-objects 
  functor $J\c\Set\to\BC$ is a left adjoint, yielding the original $T$.
  \item $TX=\realp\times X+\realpe$ is a monad, which can be used   
  for \autoref{exa:hybrid}. Here, $\realp\times X$ refers  
  to terminating behaviours and $\realpe = \realp\cup\{\infty\}$ to
  Zeno and infinite behaviours. 
\end{enumerate}
\end{example}

\section{Freyd Categories and Strong Monads}\label{sec:mon} %
\begin{figure*}[t]
\begin{equation*} 
\begin{gathered}
\frac{x\c A\text{~~in~~}\G}{~\G\ctxv x\c A~}\qquad
\frac{f\c A\to B\in\Sigv\quad\G\ctxv v\c A}{~\G\ctxv f(v)\c B~}\qquad
\frac{f\c A\to B\in\Sigc\quad\G\ctxv v\c A}{~\G\ctxc f(v)\c B~}\\[2ex]
\frac{~\G\ctxv v\c A~}{\G\ctxc \ret v\c A}\qquad
\frac{~\G\ctxc p\c A\quad \Gamma,x\c A\ctxc q\c B}{\G\ctxc \mbind{x\gets p}{q}\c B}\qquad
\frac{\G\ctxv v\c 0}{\G\ctxc\oname{init} v\c A}
\\[3ex]
\frac{~\G\ctxv v\c A~}{\G\ctxv \linj v\c A+B}\quad~~
\frac{~\G\ctxv v\c B~}{\G\ctxv \rinj v\c A+B}\quad~~
 \frac{%
   \G\ctxv v\c A+B\quad x\c A\ctxc p\c C\quad y\c B\ctxc q\c C
  }{%
   \G\ctxc\case{v}{\linj x\mto p}{\rinj y\mto q}\c C
  }\\[2ex]
\frac{%
		\G\ctxv v\c A\qquad \G\ctxv w\c B 
	}{%
		\G\ctxv \brks{v,w}\c A\times B
	}
\qquad\qquad
  \frac{%
    \G\ctxv v\c A\times B \qquad
    \G,x\c A,y\c B \ctxc q\c C
  }{%
    \G \ctxc \pcase{v}{\brks{x,y} \mto q}\c C
  }
\end{gathered}
\end{equation*}
 \caption{FGCBV with coproducts.}
  \label{fig:fgcbv}
\end{figure*}
The full FGCBV (with coproducts) is obtained by extending the type syntax~\eqref{eq:types}
with binary products $A\times B$,
and by replacing the rules in~\autoref{fig:simple} with the rules in~\autoref{fig:fgcbv}.
We now assume that variable contexts $\Gamma$ are (possibly empty) lists 
$(x_1\c A_1,\ldots,x_n\c A_n)$ with non-repetitive $x_1,\ldots,x_n$.
To interpret the resulting language, again, we need 
an identity-on-objects functor~$J\c\BV\to\BC$, an action of $\BV$
on $\BC$, and $J$ to preserve this action.
\begin{definition}[Actegory~\cite{JanelidzeKelly01}]\label{def:act}
Let $(\BV,\ten, I)$ be a monoidal category. Then an \emph{action} of~$\BV$ on a 
category $\BC$ is a bifunctor $\oslash\c\BV\times\BC\to\BC$ together with the 
\emph{unitor} and the \emph{actor} natural isomorphisms $\upsilon\c I\oslash X\iso X$,
$\alpha\c X\oslash (Y\oslash Z)\iso (X\ten Y)\oslash Z$, satisfying the following 
coherence conditions
\begin{equation*}
\qquad\begin{tikzcd}[column sep=2em, row sep=normal]
I\oslash (X\oslash Y)
	\rar["\ups"]\ar[dr,"\alpha"']& 
X\oslash Y
\\               
&(I\ten X)\oslash Y
	\uar["\iso"']                 
\end{tikzcd}\hspace{6ex}
\begin{tikzcd}[column sep=2em, row sep=normal]
X\oslash Y & 
X\oslash (I\oslash Y)
	\lar["\id\oslash\ups"']
	\ar[dl,"\alpha"]
\\               
(X\ten I)\oslash Y 
	\uar["\iso"]             
\end{tikzcd}\qquad
\end{equation*}%

\medskip
\begin{equation*}
\begin{tikzcd}[column sep=0em, row sep=normal]
 X\oslash (Y\oslash (Z\oslash V))
  \rar["\id\oslash\alpha"]
   \dar["\alpha"'] &[3em]
 X\oslash ((Y\ten Z)\oslash V)
  \arrow[r, "\alpha"] &[3ex]
  (X\ten (Y\ten Z))\oslash V
  \arrow[d, "\iso"] 
 \\
 (X\ten Y)\oslash (Z\oslash V)
  \arrow[rr, "\alpha"'] &
  & ((X\ten Y)\ten Z)\oslash V
\end{tikzcd}
\end{equation*}%
(eliding the names of canonical isomorphisms). Then $\BC$ is called an \emph{($\BV$-)actegory}.
\end{definition}
Note that every monoidal category trivially acts on itself via $\oslash=\ten$.
In the sequel, we will only consider \emph{Cartesian} categories, i.e.\ actegories
w.r.t.\ $(\BV,\times,1)$.
\begin{definition}[Freyd Category~\cite{Levy04}]\label{def:freyd}
A \emph{Freyd category} $(\BV,\BC,J(\argument),\oslash)$ 
consists of the following data:
\begin{enumerate}%
  \item a Cartesian category~$\BV$;
  \item a category $\BC$ with $|\BV|=|\BC|$;
  \item an identity-on-objects functor~$J\c\BV\to\BC$;
  \item an action of~$\BV$ on~$\BC$, such that $J$ preserves the $\BV$-action,
  i.e.\ $J(f\times g) = f\oslash Jg$ for all $f\in\BV(X,X')$, $g\in\BV(Y,Y')$ (entailing $X\times Y = X\oslash Y$ for all $X,Y\in|\BV|$),
  $\ups = J\snd$ and $\alpha = J\brks{\id\times\fst,\snd\comp\snd}$.
\end{enumerate}
\end{definition}
Let us reformulate this definition slightly more explicitly.
\begin{lemma}\label{eq:freyd_def}
A tuple $(\BV,\BC,J(\argument),\oslash)$ is a Freyd category iff
\begin{enumerate}%
  \item $\BV$ is a Cartesian category;
  \item $\BC$ is a category, such that $|\BV|=|\BC|$;
  \item $J$ is an identity-on-objects functor $\BV\to\BC$;
  \item $\oslash$ is a bifunctor $\BV\times\BC\to\BC$, such that 
  \begin{enumerate}
    \item $X\oslash Y=X\times Y$ for all $X,Y\in |\BV|$, 
    \item $J(f\times g) = f\oslash Jg$ for all $f\in\BV(X,X')$,
  $g\in\BV(Y,Y')$, and
    \item every $J\snd\c 1\times X\to X$ is natural in $X$ (w.r.t.\ $\BC$-morphisms) and every 
    $J\brks{\id\times\fst\comma\snd\comp\snd}\c X\times {(Y\times Z)}\to (X\times Y)\times Z$
    is natural in $X,Y$ (w.r.t.\ $\BV$-morphisms) and $Z$ (w.r.t.\ $\BC$-morphisms).
  \end{enumerate}
  \end{enumerate}
\end{lemma}
\begin{proof}
The claim follows from the observation that defining $\ups$ and $\alpha$ as $J\snd$ 
and $J\brks{\id\times\fst\comma\snd\comp\snd}$ correspondingly, yields an action of 
$\BV$ on $\BC$ iff $\ups$ and $\alpha$ are natural -- the coherence conditions 
from \autoref{def:act} hold automatically.
\end{proof}

\begin{definition}[Distributive Freyd Category~\cite{Staton14}]
A Freyd category $(\BV,\BC,J(\argument),\oslash)$ is~\emph{distributive} if\/ $\BV$ is distributive, 
$\BC$ is co-Cartesian, and~$J$ strictly preserves coproducts.
\end{definition}
Note that it follows from the above definition that the action $\oslash$ preserves
coproducts in the second argument. Indeed, by applying $J$ to the isomorphism 
$X\times (Y+Z)\iso X\times Y+X\times Z$ in $\BV$, we obtain that $X\oslash (Y+Z)\iso X\oslash Y+X\oslash Z$
is in $\BC$.

Given a distributive Freyd category $(\BV,\BC,J(\argument),\oslash)$, we update the semantics
from~\autoref{sec:simp} by extending the semantics of types with the clauses $\sem{A\times B} = 
\sem{A}\times\sem{B}$, $\sem{x_1\c A_1\comma\ldots\comma x_n\c A_n}=\sem{A_1}\times\ldots\times\sem{A_n}$,
and by defining the semantics of terms as in~\autoref{fig:fgcbv-sem}, where 
$\proj_i\c X_1\times\ldots\times X_n\to X_i$ denotes the $i$-th projection.
\begin{figure*}[t]
\begin{gather*}
\frac{}{\sem{x_1\c A_1,\ldots,x_n\c A_n\ctxv x_i\c A_i} = \proj_i} 
\\[2ex]
\frac{h=\sem{\G\ctxv v\c A}}{\sem{\G\ctxv f(v)\c B} = \sem{f}\comp h}
\qquad\quad  
\frac{h=\sem{\G\ctxv v\c A}}{\sem{\G\ctxc f(v)\c B} = \sem{f}\comp Jh}
\\[2ex] 
\frac{h=\sem{\G\ctxv v\c A}}{\sem{\G\ctxc \ret v\c A} = Jh}
\qquad\quad
\frac{h_1=\sem{\G,x\c A\ctxc q\c B}\qquad h_2=\sem{\G\ctxc p\c A}}{\sem{\G\ctxc \mbind{x\gets p}{q}\c B} = h_1\comp (\id\oslash h_2)\comp J\Delta}
\\[2ex]
\frac{}{\sem{\G\ctxc\oname{init} v\c A} = \bang}
\qquad
\frac{h=\sem{\G\ctxv v\c A}}{\sem{\G\ctxv\linj v\c A+B} = \inl\comp h}
\qquad
\frac{h=\sem{\G\ctxv v\c B}}{\sem{\G\ctxv\rinj v\c A+B} = \inr\comp h}
\\[2ex]
\frac{h=\sem{\G\ctxv v\c A+B}\qquad h_1=\sem{\G,x\c A\ctxc p\c C}\qquad h_2=\sem{\G,y\c B\ctxc q\c C}}{\sem{\G\ctxc\case{v}{\linj x\mto p}{\rinj y\mto q}\c C} = [h_1,h_2]\comp J\dist\comp (\id\oslash Jh)\comp J\Delta}
\\[2ex]
\frac{
		h_1=\sem{\G\ctxv v\c A}\qquad 
		h_2=\sem{\G\ctxv w\c B} 
}{
	\sem{\G\ctxv \brks{v,w}\c A\times B} = \brks{h_1,h_2}
}
\\[2ex]
  \frac{%
    h_1=\sem{\G\ctxv p\c A\times B} \qquad
    h_2=\sem{\G,x\c A,y\c B \ctxc q\c C}
  }{%
    \sem{\G \ctxc \pcase{p}{\brks{x,y} \mto q}\c C} = h_2\comp (\id\oslash Jh_1)\comp J\Delta
  }
\end{gather*}
 \caption{Denotational semantics of FGCBV with coproducts.}
  \label{fig:fgcbv-sem}
\end{figure*}
Freyd categories are to strong monads as identity-on-objects functors to monads.
\begin{proposition}\label{prop:freyd}%
Let~$(\BV,\BC,J(\argument),\oslash)$ be a Freyd category. Then~$J$ is a left adjoint 
iff~$\BC$ is isomorphic to a Kleisli category of some strong monad\/~$\BBT$ on~$\BV$ and 
$J f = H(\eta\comp f)$ for all~$f\in\BV(X,Y)$ where~$H\c\BV_\BBT\iso\BC$ is the 
relevant isomorphism.
\end{proposition}
\begin{proof}
Freyd categories are initially designed to generalize Kleisli categories of
strong monads~\cite{PowerThielecke99}, in particular, we obtain the \emph{`If'} direction of
the claim.  

For the \emph{`Only if'} direction, suppose that $J$ is a left adjoint, and show that 
the requested strong monad exists. Indeed, we obtain a monad $\BBT$ by \autoref{prop:kleisli}.
W.l.o.g.\ suppose that $\BC=\BV_{\BBT}$. Let $Jf=\eta\comp f$ for every $f\c X\to Y$
from $\BV$. Let us define strength $\tau$ as $\id_X\oslash\id_{TY}\c X\times TY\to T(X\times Y)$,
which is clearly natural in $X$ and~$Y$. If follows that $f\oslash g = \tau\comp (f\times g)$.
Indeed, 
\begin{align*}
f\oslash g 
	=&\; f\oslash(\id^\klstar\comp\eta\comp g)\\*
	=&\; (\id\oslash\id)^\klstar\comp (f\oslash \eta\comp g)\\
	=&\; (\id\oslash\id)^\klstar\comp (f\oslash Jg)\\
	=&\; (\id\oslash\id)^\klstar\comp J(f\times g)\\
	=&\; \tau^\klstar\comp \eta\comp(f\times g)\\
	=&\; \tau\comp (f\times g).
\end{align*}
Let 
$\gamma=\brks{\id\times\fst,\snd\comp\snd}\c X\times (Y\times Z)\iso (X\times Y)\times Z$. 
The axioms of strength are verified as follows.
\begin{enumerate}
  \item Using \autoref{eq:freyd_def} (4.c): $(T\snd)\comp\tau = (J\snd)^\klstar\comp (\id\oslash\id) = \id^\klstar\comp J\snd= \id^\klstar\comp \eta\comp\snd = \snd$.
  \item Using \autoref{eq:freyd_def} (4.c): $T\gamma\comp \tau\comp (\id\times\tau) = (J\gamma)^\klstar\comp (\id\oslash(\id\oslash\id)) = ((\id\times\id)\oslash\id)^\klstar\comp J\gamma = \tau^\klstar\comp\eta\comp\gamma = \tau\comp\gamma$.
  \item $\tau\comp(\id\times\eta) = \tau^\klstar\comp\eta\comp(\id\times\eta) = (\id\oslash\id)^\klstar\comp J(\id\times\eta) = (\id\oslash\id)^\klstar\comp (\id\oslash J\eta) = \id\oslash (\id^\klstar\comp J\eta) 
  = \id\oslash (\id^\klstar\comp\eta\comp\eta) =\id\oslash\eta =\id\oslash J\id = J(\id\times\id) = \eta$.
  \item $(\tau\comp (f\times g))^\klstar\comp\tau = (f\oslash g)^\klstar\comp (\id\oslash\id) = f\oslash g^\klstar = \tau\comp (f\times g^\klstar)$.\qed
\end{enumerate}
\noqed\end{proof}
\autoref{prop:freyd} allows us to refactor the existing characterization of
\emph{closed Freyd categories}~\cite[Theorem 7.3]{LevyPowerEtAl02} along the following lines. 
In order to include higher-order types in the language, we would need to add~${A\to B}$ as a 
new type former and the following term formation rules:
\begin{flalign*}
&&
\frac{%
	\G, x\c A\ctxc p\c B
}{%
	\Gamma \ctxv \lambda x.\, p\c A\to B
}
&&
\frac{%
  \Gamma \ctxv w\c A\qquad \Gamma \ctxv v\c A\to B
}{%
  \G \ctxc v w\c B
}
&&
\end{flalign*}
We would then need to provide the following additional semantic clauses:
\begin{gather*}
\frac{h=\sem{\G,x\c A \ctxc p\c B}}{\sem{\Gamma \ctxv \lambda x.\, p\c A\to B} = \curry h}
\qquad\quad
\frac{h_1=\sem{\Gamma \ctxv v\c A\to B}\qquad h_2=\sem{\Gamma \ctxv w\c A}}{\sem{\Gamma \ctxc v w\c B} = (\curry^\mone h_1)\comp (\id\oslash Jh_2)\comp J\Delta}
\end{gather*}
where $\sem{A\to B} = \sem{A}\multimap\sem{B}$, $\multimap\c |\BV|\times|\BC|\to|\BC|$,
and $\curry$ is an isomorphism 
\begin{align}\label{eq:curry}
\curry\c\BC(J(X\times A),B)\iso\BV(X,A\multimap B)
\end{align}
natural in $X$. In particular, this says that $J$ is left adjoint to $1\multimap (\argument)$, 
which, as we have seen in~\autoref{prop:kleisli}, means that $\BC$ is isomorphic to the Kleisli category of a strong monad~$\BBT$,
and hence~\eqref{eq:curry} amounts to $\BV(X\times A,TB)\iso\BV(X,A\multimap B)$, i.e.\ to 
the existence of \emph{Kleisli exponentials}, which are exponentials of the form $(TB)^{A}$.
We thus obtain the following 
\begin{corollary}\label{cor:str-freyd}
Let~$(\BV,\BC,J(\argument),\oslash)$ be a Freyd category. The following are equivalent:
\begin{itemize}%
  \item an isomorphism~\eqref{eq:curry} natural in $X$ exists;
  \item for all $A\in |\BV|$, $J(\argument\times A)\c\BV\to\BC$ is a left adjoint;
  \item $\BC$ is isomorphic to a Kleisli category of a strong monad, and 
  Kleisli exponentials exist.
\end{itemize}
\end{corollary}
A yet another way to express~\eqref{eq:curry} is to state that the presheaves
\begin{displaymath}
	\BC(J(\argument\times A),B)\c\BV^\op\to\Set
\end{displaymath} 
are representable. We will use this 
formulation in our subsequent analysis of guardedness.

\section{Guarded Freyd Categories}\label{sec:gfreyd}
We proceed to recall the formal notion of guardedness~\cite{GoncharovSchroderEtAl17,LevyGoncharov19}.
\begin{definition}[Guardedness]\label{def:grd}
A \emph{guardedness predicate} on a co-Cartesian category~$\BC$ provides for all 
$X\comma Y\comma Z\in |\BC|$ a subset~$\GHom*{\BC}(X,Y,Z) \subseteq \BC(X, Y+Z)$,
whose elements we write as $f \c X\to Y \gsep Z$ and call guarded (in~$Z$), such that 
{\allowdisplaybreaks[0]\upshape
\begin{gather*}
\textbf{(trv$_\gtag$)}\quad\frac{f\c X\to Y}{~\dcomp{f}{\inl}\c X\to Y\gsep Z~}\qquad
\textbf{(par$_\gtag$)}\quad\frac{~f\c X\to V\gsep W
\qquad~g\c Y\to V\gsep W}
{~[f,g]\c X+Y\to V\gsep W}\\[2ex]
\qquad\quad\textbf{(cmp$_\gtag$)}\quad\frac{~f\c X\to Y\gsep Z\qquad g\c Y\to V\gsep W\qquad h\c Z\to V+W}{\dcomp{f}{[g,h]}\c X\to V\gsep W} %
\end{gather*}
}

\noindent A \emph{guarded (co-Cartesian) category} is a category equipped with a 
guardedness predicate. 
A \emph{guarded functor} between two guarded categories 
is a functor $F\c\BC\to\BD$ that strictly preserves coproducts, and preserves 
guardedness in the following sense: $f\in\GHom*{\BC}(X,Y,Z)$ entails 
$f\in\GHom*{\BD}(FX,FY,FZ)$. 
\end{definition}
It follows from the axioms of guardedness that $\GHom*{\BC}$ is a functorial operator.
\begin{proposition}\label{pro:guard-funct}
$\GHom*{\BC}$ extends to a functor $\BC^\op\times\BC\times\BC\to\Set$.
\end{proposition}
\begin{proof}
The map $X,Y,Z\mapsto \BC(X,Y+Z)$ is obviously functorial. We are left to check 
that given $f\c X\to Y\gsep Z$, $g\c X'\to X$, $h\c Y\to Y'$ and $u\c Z\to Z'$,
$(h+u)\comp f\comp g\c X'\to Y'\gsep Z'$. Observe first that $f\comp g\c X'\to Y\gsep Z$.
Indeed, $f\comp g = [f,f]\comp \inl\comp g\to Y\gsep Z$ using~\textbf{(trv$_\gtag$)} 
and~\textbf{(cmp$_\gtag$)}. Next, again by~\textbf{(cmp$_\gtag$)}, $[\inr\comp h,\inl\comp u]\comp f\comp g = (h+u)\comp f\comp g\c X'\to Y'\gsep Z'$.
\end{proof} 
In the sequel, we regard $\gsep$ as an operator that binds the weakest. 
Intuitively, $\GHom*{\BC}(X,Y,Z)$ axiomatically distinguishes those morphisms $X\to Y+Z$ for which 
the program flow from~$X$ to $Z$ is guarded, in particular, if $X=Z$ then the corresponding guarded loop 
can be safely closed. Note that the standard (totally defined) iteration is an 
instance with $\GHom*{\BC}(X,Y,Z) = \BC(X,Y+Z)$. Consider other instances.
\begin{example}[Vacuous Guardedness~\cite{GoncharovSchroder18}]
The least guardedness predicate is as follows: $\GHom*{\BC}(X,Y,Z) = \{\inl\comp f\c X\to Y + Z\mid f\in\BC(X, Y)\}$. 
Such~$\BC$ is called \emph{vacuously guarded}.
\end{example}
The following class of examples abstracts the monad of synchronization trees from~\autoref{exa:resump}: 
$\BBT$ can capture arbitrary ``branching'' computational effects besides
$T=\FSet$ for nondeterminism, and $H$ can capture arbitrary ``action'' functors besides $HX = A\times X$
for standard process algebra actions.   
\begin{example}[Coalgebraic Resumptions]\label{exa:resump}
Let~$\BBT$ be a monad on a co-Cartesian category~$\BV$, and let
$H\c\BV\to\BV$ be an endofunctor such that all fixpoints~$T_H
X = \nu \gamma.\,T(X + H \gamma)$ exist. These jointly yield a monad
$\BBT_{H}$, called the \emph{(generalized) coalgebraic resumption monad
(transform of~$\BBT$)}~\cite{PirogGibbons14,GoncharovSchroderEtAl17}. Then the Kleisli category of
$\BBT_{H}$ is guarded with~$f\c X\to {Y\gsep Z}$ if 
\begin{equation}\label{eq:resump-guard}
\begin{tikzcd}[column sep = 10ex,row sep = 3ex]
X
  \ar[r, "g"] 	
  \ar[d, "f"'] & 
T(Y+H T_{H}(Y+Z))
  \ar[d,"T(\inl+\id)"]\\
T_{H}(Y+Z)
  \ar[r, "\out"] & 
T((Y+Z)+HT_{H}(Y+Z))
\end{tikzcd}
\end{equation}
for some~$g\c X\to T(Y+HT_{H}(Y+Z))$. 
Guarded iteration operators canonically extend from $\BBT$ to~$\BBT_H$~\cite{LevyGoncharov19}.
\end{example}
The next example is interesting in that the notion of guardedness is defined
essentially the same way, but fixpoints of guarded morphisms need not exist. 
\begin{example}[Algebraic Resumptions]\label{exa:fin-resump}
A simple variation of the previous example involves least fixpoints 
$T^HX=\mu\gamma.\,T(X + H \gamma)$ instead of the greatest ones,
and $\inm^\mone$ instead of~$\out$, where $\inm\c T(X+HT^HX)\to T^HX$ is the initial algebra structure of $T^HX$,
which is an isomorphism by Lambek's lemma.
However, we can no longer generally induce non-trivial (guarded) iteration operators for $\BBT^H$.\sgnote{Add a diagram.}
\end{example}
\begin{example}\label{exa:guard}
Let us describe natural guardedness predicates on the Kleisli categories of monads 
from \autoref{exa:monad}.
\begin{enumerate}[wide]
  \item $TX=\nu\gamma.\,\FSet((X+1)+A\times\gamma)$ is a special case of 
  \autoref{exa:resump}. The 
  guardedness condition~\eqref{eq:resump-guard} 
  instantiates as follows: $f\c X\to \nu\gamma.\,\FSet((Y+Z+1)+A\times\gamma)$ is 
  guarded if $\out\comp f\c X\to \FSet((Y+Z+1)+A\times T(Y+Z))$ factors through $\FSet((Y+1)+A\times T(Y+Z))$,
  i.e.\ the only allowed way to terminate through $Z$ is that which is preceded by an action from $A$.
  \item\label{it:guard2} For $TX=\PSet(A^\star\times (X+1) + A^\star)$, let $f\c X\to Y\gsep Z$
  if for every $x\in X$, $\inl (w\comma\inl(\inr y))\in f(x)$ entails~$w\neq\eps$. 
  \item For $TX=\PSet(A^\star\times (X+1) + (A^\star+A^\omega))$ guardedness 
  is defined as in clause~\eqref{it:guard2}. %
  \item For $TX=(\nu\gamma.\,X\times S+\gamma\times S)^S$, recall that $\Set_{\BBT}$
  is isomorphic to a full subcategory of the Kleisli category of $\nu\gamma.\,(\argument+\gamma\times S)$,
  which is again an instance of \autoref{exa:resump} with $TX = X$ and $HX=X\times S$.
  The guardedness predicate for~$\BBT$ thus restricts accordingly.
  \item For $TX=(\realp\times X)+\realpe$ let $f\c X\to Y\gsep Z$ if~$f(x) = \inl\, (r,\inr z)$ 
implies~$r>0$. 
\end{enumerate}
\end{example}
We proceed to extend the language in~\autoref{fig:fgcbv} with guardedness data. 
As before,~$\Sigv$ consists of constructs of the 
form~$f\c A\to B$, while~$\Sigc$ consists of constructs of the form $f\c A\to B\gsep C$,
indicating guardedness in $C$.
The new formation rules are then given in~\autoref{fig:gfgcbv}.
The rule for $\ret$ now introduces a coproduct summand $B$ with respect
to which the computation is vacuously guarded, thus adhering to~\textbf{(trv$_\gtag$)}. 
The rule for binding now must incorporate~\textbf{(cmp$_\gtag$)}, which requires 
the following modification of the syntax:\sgnote{recover the original one.}
\begin{displaymath}
  \docase{p}{\linj x\mto q}{\rinj y\mto r}
\end{displaymath}
The latter construct is meant to be equivalent to $\mbind{z\gets p}{\case{z}{\linj x\mto q}{\rinj y\mto r}}$.
modulo guardedness information.
Finally,~\textbf{(par$_\gtag$)} is captured by the formation rule for~$\oname{case}$, which is essentially unchanged w.r.t.~\autoref{fig:fgcbv}.
An analogue of the iteration operator~\eqref{eq:iter} in the new setting would be the rule:
\begin{align*}
  \frac{\G\ctxc p\c A\gsep\iobj\qquad \G,x\c A\ctxc q\c B\gsep C+A}{\G\ctxc\ibind{x\gets p}{q}\c B\gsep C}
\end{align*}
\begin{figure*}[t]
\begin{equation*} 
\begin{gathered}
\frac{x\c A\text{~~in~~}\G}{~\G\ctxv x\c A~}\qquad
\frac{f\c A\to B\in\Sigv\quad\G\ctxv v\c A}{~\G\ctxv f(v)\c B~}\qquad
\frac{f\c A\to B\gsep C\in\Sigc\qquad\G\ctxv v\c A}{~\G\ctxc f(v)\c B\gsep C~}\\[2ex]
\frac{~\G\ctxv v\c A~}{\G\ctxc \ret v\c A\gsep B}\quad~~
\frac{~\G\ctxc p\c A\gsep B\quad~ \G,x\c A\ctxc q\c C\gsep D\quad~ \G,y\c B\ctxc r\c C+D\gsep 0}{\G\ctxc\docase{p}{\linj x\mto q}{\rinj y\mto r}\c C\gsep D}\\[2ex]
\frac{\G\ctxv v\c 0}{\G\ctxc\oname{init} v\c A}\qquad\quad
\frac{~\G\ctxv v\c A~}{\G\ctxv \linj v\c A+B}\qquad\quad
\frac{~\G\ctxv v\c B~}{\G\ctxv \rinj v\c A+B}\\[2ex]
 \frac{%
   \G\ctxv v\c A+B\qquad \G,x\c A\ctxc p\c C\gsep D\qquad \G,y\c B\ctxc q\c C\gsep D
  }{%
   \G\ctxc\case{v}{\linj x\mto p}{\rinj y\mto q}\c C\gsep D
  }
\\[2ex]  
\frac{%
		\G\ctxv v\c A\qquad \G\ctxv w\c B 
	}{%
		\G\ctxv \brks{v,w}\c A\times B
	}\qquad
\frac{%
  \G\ctxv p\c A\times B \qquad
  \G,x\c A,y\c B \ctxc q\c C\gsep D
}{%
  \G \ctxc \pcase{p}{\brks{x,y} \mto q}\c C\gsep D
}
\end{gathered}
\end{equation*}
 \caption{Term formation rules of guarded FGCBV.}
  \label{fig:gfgcbv}
\end{figure*}
\begin{example}[Weakening]
One can expect that the judgement $f\c X\to Y\gsep Z+W$ entails $f\c X\to Y+Z\gsep W$, 
meaning that if a morphism is guarded w.r.t.\ an object $Z+W$, then it is guarded
w.r.t.\ to its part $W$. The corresponding \emph{weakening} principle 
\begin{gather*}
\textbf{(wkn$_\gtag$)}\quad\frac{f\c X\to Y\gsep Z+W}{~f\c X\to Y+Z\gsep W~}\qquad
\end{gather*}
is indeed derivable from~\textbf{(trv$_\gtag$)},~\textbf{(par$_\gtag$)} and ~\textbf{(cmp$_\gtag$)}.
In terms of guarded FGCBV, this corresponds to constructing the following term
from a given $\G\ctxc p\c A\gsep B+C$:
\begin{align*}
  \G\ctxc\docase{p}{\linj x\mto&\;\ret(\linj x)}{\\*\rinj z\mto&\;\case{z}{\linj x\mto\ret (\linj(\rinj x)) }{\\*
  &\;\hspace{3.6em}\rinj y\mto\ret (\rinj y)}}\c A+B\gsep C.
\end{align*}
\end{example}
\begin{example}
The updated effectful signature of \autoref{exa:bpl} now involves $a\c {1\to 0\gsep 1}$
and $\oname{toss}\c 1\to 2\gsep 0$, indicating that actions guard everything, 
while nondeterminism guards nothing. 
The signature $\Sigc$ from \autoref{exa:itrace} can be refined to
$\{\oname{put}\c S\to 0\gsep 1,\oname{get}\c 1\to S\gsep 0\}$, meaning again that $\oname{put}$ guards everything and $\oname{get}$ guards nothing. 
\autoref{exa:hybrid} is more subtle since $\oname{wait}\c{\realp\to\realp}$ 
is meant to be guarded only for non-zero inputs. We thus can embed the involved 
case distinction into $\oname{wait}$ by redefining it as $\oname{wait}\c{\realp\to\realp\gsep\realp}$.
\end{example}

\begin{definition}[Guarded Freyd Category]
A distributive Freyd category $(\BV,\BC,J(\argument),\oslash)$ is guarded
if $\BC$ is guarded and the action of\/ $\BV$ on $\BC$ preserves guardedness in the following sense:
Given $f\in\BV(A,B)$, $g\in\GHom*{\BC}(X,Y,Z)$, $J\dist\comp (f\oslash g)\in\GHom*{\BC}(A\times X, B\times Y, B\times Z)$.
\end{definition}
The semantics of~$(\Sigv,\Sigc)$ over a guarded Freyd category $(\BV,\BC,J(\argument),\oslash)$ 
interprets types and operations from $\Sigv$ as before and 
sends each~$f\c A\to B\gsep C\in\Sigc$ to $\sem{f}\in\GHom*{\BC}(\sem{A},\sem{B},\sem{C})$.
Terms in context are now interpreted as $\sem{\G\ctxv v\c B}\in\BV(\sem{\G},\sem{B})$
and $\sem{\G\ctxc p\c B\gsep C}\in\BC(\sem{\G},\sem{B}+\sem{C})$, according to the 
rules in \autoref{fig:gfgcbv-sem}.
\begin{figure*}
\begin{gather*}
\frac{}{\sem{x_1\c A_1,\ldots,x_n\c A_n\ctxv x_i\c A_i} = \proj_i}
\\[2ex]
\frac{h=\sem{\G\ctxv v\c A}}{\sem{\G\ctxv f(v)\c B} = \sem{f}\comp h}\qquad
\frac{h=\sem{\G\ctxv v\c A}}{\sem{\G\ctxc f(v)\c B\gsep C} = \sem{f}\comp J h}
\\[2ex]
\frac{h=\sem{\G\ctxv v\c A}}{\sem{\G\ctxc \ret v\c A\gsep B} = J\inl\comp Jh}
\\[2ex]
\frac{h = \sem{\G\ctxc p\c A\gsep B}\quad~~ h_1=\sem{\G,x\c A\ctxc q\c C\gsep D}\quad~~ h_2=\sem{\G,y\c B\ctxc r\c C+D\gsep \iobj}}{\sem{\G\ctxc\docase{p}{\linj x\mto q}{\rinj y\mto r}\c C\gsep D} =[h_1, [\id,\bang]\comp h_2]\comp J\dist\comp (\id\oslash h)\comp J\Delta}
\\[2ex]
\frac{}{\sem{\G\ctxc\oname{init} v\c A} = \bang}
\qquad
\frac{h=\sem{\G\ctxv v\c A}}{\sem{\G\ctxv\linj v\c A+B} = \inl\comp h}\qquad
\frac{h=\sem{\G\ctxv v\c B}}{\sem{\G\ctxv\rinj v\c A+B} = \inr\comp h}
\\[2ex]
\frac{h=\sem{\G\ctxv v\c A+B}\quad h_1=\sem{\G,x\c A\ctxc p\c C\gsep D}\quad h_2=\sem{\G,y\c B\ctxc q\c C\gsep D}}{\sem{\G\ctxc\case{v}{\linj x\mto p}{\rinj y\mto q}\c C\gsep D} = [h_1, h_2]\comp J\dist\comp (\id\oslash J h)\comp J\Delta}
\\[2ex]
\frac{h_1=\sem{\G\ctxv v\c A}\qquad h_2=\sem{\G\ctxv w\c B}}{\sem{\G\ctxv \brks{v,w}\c A\times B} = \brks{h_1,h_2}}
\end{gather*}
 \caption{Denotational semantics of guarded FGCBV over guarded Freyd categories.}
  \label{fig:gfgcbv-sem}
\end{figure*}
This is well-defined, which can be easily shown 
by structural induction:
\begin{proposition}
For any derivable $\G\ctxc p\c A\gsep B$, 
$\sem{\G\ctxc p\c A\gsep B}\in\GHom*{\BC}(\sem{\G},\sem{A},\sem{B})$.
\end{proposition}
\section{Representing Guardedness}\label{sec:repr}
In \autoref{sec:mon} we explored the combination of strength (i.e., multivariable contexts) 
and the representability of presheaves $\BC(J(\argument),X)\c{\BV^\op\to\Set}$, sticking to 
the bottom face of the cube in~\autoref{fig:cube}. 
Our plan now is to obtain additional concepts by examining 
the representability of~$\GHom*{\BC}(J(\argument),X,Y)\c\BV^\op\to\Set$.
Note that representability of guardedness together with function spaces amounts to representability 
of~$\GHom*{\BC}(J(\argument\times X),Y,Z)\c\BV^\op\to\Set$, i.e.\ to the
existence of an endofunctor $\multimap\c\BV^\op\times\BC\times\BC\to\BC$, such that
$\GHom*{\BC}(J(\argument\times X),Y,Z)\iso\BV(\argument,X\multimap_Z Y)$.
This is exactly the structure one would need to extend~\autoref{fig:gfgcbv}
with function spaces as follows:
\begin{flalign*}
&&
\frac{%
		\G, x\c A\ctxc p\c B\gsep C
	}{%
		\Gamma \ctxv \lambda x.\, p\c A\to_C B
	}
&&
  \frac{%
    \Gamma \ctxv w\c A\qquad \Gamma \ctxv v\c A\to_C B
  }{%
    \G \ctxc v w\c B\gsep C
  }
&&
\end{flalign*}
The decorated function spaces $A\to_C B$ can then be interpreted as 
${\sem{A}\multimap_{\sem{C}} \sem{B}}$, which is a subobject of the 
Kleisli exponential $\sem{A}\to T(\sem{B}+\sem{C})$, consisting of guarded morphisms.
\begin{definition}
Given~$J\c\BV\to\BC$, where~$\BC$ is guarded, we call the guardedness 
predicate~$\GHom*{\BC}$~$J$-representable if for all~$X,Y\in |\BC|$ the presheaf 
$\GHom*{\BC}(J(\argument), X, Y)\c\BV^{\op}\to\Set$
is representable, i.e.\ for all~$X,Y\in |\BC|$ there is ${U(X,Y)\in|\BV|}$ such that
\begin{align}\label{eq:ghom-iso}
\GHom*{\BC}(JZ,X,Y) \iso \BV(Z,U(X,Y)) 
\end{align} 
naturally in~$Z$. A guardedness predicate $\GHom*{\BC}$ is called
$J$-guarded if it is equipped with a~$J$-representable guardedness predicate.
\end{definition}
\begin{lemma}\label{lem:cof}
Given an identity-on-objects functor~${J\c\BV\to\BC}$, 
$\GHom*{\BC}$ is~$J$-representable~iff 
\begin{itemize}
  \item there is a family of objects $(U(X,Y)\in |\BV|)_{X,Y\in |\BC|}$;
  \item there is a family of guarded morphisms $(\eps_{X,Y}\c U(X,Y)\to X\gsep Y)_{X,Y\in |\BC|}$;
  \item there is an operator~$(\argument)^\natural\c\GHom*{\BC}(Z,X,Y)\to\BV(Z,U(X,Y))$ sending 
   each~${f\c Z\to X\gsep Y}$ to the unique morphism $f^\natural$
    for which 
the diagram
\begin{equation*}
\begin{tikzcd}[column sep = 12ex,row sep = 3ex]
& U(X,Y)
  \ar[d,"\eps_{X,Y}"]\\
Z
	\ar[ur, "Jf^\natural"] 	
	\ar[r, "f"'] & X + Y
\end{tikzcd}
\end{equation*}
commutes.
\end{itemize}
These conditions entail that $U$ is a bifunctor and that $\eps_{X,Y}$ is natural 
in $X$ and $Y$.
\end{lemma}
\begin{proof}
First, we argue that the declared characterization entails that $U$ is a bifunctor
and $\eps_{X,Y}$ is natural in $X$ and $Y$. Let~$g\c X\to X'$, 
$h\c Y\to Y'$, and note that~$(g+h)\comp\eps_{X,Y}\c JU(X,Y)\to X'\gsep Y'$. 
Then the diagram 
\begin{equation*}
\begin{tikzcd}[column sep = 12ex,row sep = 4ex]
JU(X,Y)
  \ar[d,"\eps_{X,Y}"']
  \ar[r, "{JU(g,h)}", dotted] &
JU(X',Y') 
  \ar[d,"\eps_{X',Y'}"] \\
X + Y
  \ar[r,"g+h"]&
X' + Y' 
\end{tikzcd}
\end{equation*}
commutes for some~$U(g,h)\c U(X,Y)\to U(X',Y')$, uniquely determined by~$g$ and 
$h$. The fact that thus defined~$U(\argument,\argument)$ is functorial is obvious 
by definition. Moreover, the above diagram establishes the naturality of~$\eps_{X,Y}$
in~$X$ and~$Y$.

Observe that, by Yoneda lemma, for any bifunctor~$U$, a natural transformation~$\xi\c\BV(\argument\comma U(X,Y)) \to \GHom*{\BC}(J(\argument),X,Y)$
is uniquely determined by a morphism~$\eps_{X,Y}\c UJ(X,Y)\to X\gsep Y$. We proceed to show that
componentwise isomorphic~$\xi$ correspond to those~$\eps_{X,Y}$ for which the 
above-described maps~$(\argument)^\natural$ exist. Note that~$\xi$ and~$\eps$
are connected as follows: 
\begin{align*}
\eps_{X,Y} =&\; \xi_{U(X,Y)}(\id\c U(X,Y)\to U(X,Y)),\\*
\xi_X(f\c Z\to U(X,Y)) =&\; \eps_{X,Y}\comp Jf. 
\end{align*}
The map~$\xi_X\c\BV(Z,U(X,Y)) \to \GHom*{\BC}(JZ,X,Y)$ is a bijection iff 
every~$f\c JZ\to JU(X,Y)$ is of the form~$\eps_{X,Y}\comp Jg$ for some 
$g\c Z\to U(X,Y)$, which is uniquely identified by~$f$, in other words, for 
every~$f\c JZ\to JU(X,Y)$ there is a unique~$f^\natural\c Z\to U(X,Y)$, such that
$f = \eps_{X,Y}\comp Jf^\natural$. 
\end{proof}
\begin{lemma}\label{lem:repr-to-adj}
If\/~$\BC$ is~$J$-guarded, then~${J\dashv U(\argument,\iobj)}$ with
$U$ as in~\autoref{lem:cof}.
\end{lemma}
\begin{proof}
Suppose that~$\BC$ is~$J$-representable with~$J\c\BV\to\BC$. Observe 
that~$\GHom*{\BC}(JX,A,\iobj)$ is isomorphic to~$\BC(JX, A)$ naturally in~$X$:
the components of the isomorphism are the maps~$f\mto\inl\comp f$ for 
$f\in\BC(JX, A)$ and~$g\mto [\id,\bang]\comp g$ for~$g\in\GHom*{\BC}(JX, A,\iobj)$.
Using~\eqref{eq:ghom-iso}, we thus arrive at
\begin{align*}
\BV(X,U(A,\iobj)) \iso \GHom*{\BC}(JX, A,\iobj)\iso \BC(JX, A),
\end{align*}    
i.e.~${J\dashv U(\argument,\iobj)}$.
\end{proof}
By~\autoref{lem:repr-to-adj}, representability fails already if $J$ has no right 
adjoint. Instructive examples of non-representability are thus only those where~$J$ \emph{does} have a right adjoint.
\begin{proposition}\label{prop:rep-in-set}
Let $\BBT$ be a monad over the category of sets $\Set$ with
the axiom of choice. If\/ $\Set_{\BBT}$ is guarded, the guardedness predicate 
is representable iff every $f\c X\to T(Y+Z)$ is guarded whenever all 
the compositions $1\ito X\xto{f} T(Y+Z)$ are guarded.
\end{proposition}
\begin{proof}
It follows from previous results~\cite[Proposition~12]{GoncharovRauchEtAl21} that 
in any category, where every morphism admits an image factorization (specifically in
$\Set$), representability of guardedness in the Kleisli category of a monad $\BBT$ is 
equivalent to the following conditions. 
\begin{enumerate}[wide]
 \item for all sets $X$ and $Y$, there is a greatest subobject $Z\ito T(X+Y)$, which is 
 guarded as a morphism;
 \item for every regular epic $e\c X'\to X$ and every morphism $f\c X\to T(Y+Z)$,
$f\comp e\c X'\to Y\gsep Z$ implies $f\c X\to Y\gsep Z$.
\end{enumerate}
The second clause follows for $\Set$: by the axiom of choice, every $e\c X'\twoheadrightarrow X$
has a section, say $m$, and then $f\comp e\c X'\to Y\gsep Z$ implies $f = f\comp e\comp m\c X'\to Y\gsep Z$.
The second clause is equivalent to the property that the injection $\bigcup_{Z\ito X\gsep Y} Z\ito T(X+Y)$ 
is guarded (and hence is the largest guarded subobject by construction).
Precomposing this map with any map whose source is $1$ yields a guarded map 
by definition, hence the condition from the proposition's statement 
is sufficient. Let us show that it is necessary. Let $f\c X\to T(Y+Z)$, and suppose 
that all the compositions $1\ito X\xto{f} T(Y+Z)$ are guarded. If the largest
guarded subobject exists, it must be the union of all such maps. Since this 
union is precisely the original map $f$, it is guarded. 
\end{proof}
\begin{example}[Failure of Representability]
In $\Set$, let $f\c X\to Y+Z$ be guarded in $Z$ if ${\{z\in Z\mid f^\mone(\inr z)\neq\emptyset\}}$ is finite.
The axioms of guardedness are easy to verify.
By~\autoref{prop:rep-in-set}, this predicate is not $\Id$-representable, 
as any $1\ito X\xto{\inr} \iobj + X$ 
is guarded, but $\inr$ is not, unless $X$ is finite.
\end{example}
In what follows, we will use $\pm$ as a binary operation that binds stronger
than monoidal products ($\ten$, $+$, $\ldots$), so, e.g.\ $X\ten Y\pm Z$ will read as 
$X\ten (Y\pm Z)$.
\begin{theorem}\label{thm:gpm}
Given an identity-on-objects guarded $J\c\BV\to\BC$, 
$\GHom*{\BC}$ is $J$-representable iff
  \begin{itemize}
    \item there is a bifunctor $\IB{}{}\c\BV\times\BV\to\BV$, such that $\argument\pm 0$ is a monad and $\BC\iso\BV_{\argument\pm 0}$;
      \item\label{it:gpm2} there is a family of guarded 
  morphisms (w.r.t.\ the guardedness predicate, induced by $\BC\iso\BV_{\argument\pm 0}$) $(\eps_{X,Y}\c\IB{X}{Y}\to X\gsep Y)_{X,Y\in |\BV|}$, natural in $X$ and $Y$;
  \item\label{it:gpm3} for every guarded $f\c X\to Y\gsep Z$, there is
   unique $f^\sharp\c X\to \IB{Y}{Z}$, such that %
   the diagram 
  \begin{equation*} %
\begin{tikzcd}[column sep = 12ex,row sep = 3ex]
   & \IB{Y}{Z}\dar["\eps_{Y,Z}"]\\
  X\ar[ur, "f^\sharp"]\ar[r,"f"'] & \IB{(Y+Z)}{\iobj}               
  \end{tikzcd}
  \end{equation*}
  commutes.
\end{itemize}
\end{theorem}
\begin{proof}
($\Rightarrow$) By \autoref{lem:repr-to-adj} and \autoref{prop:kleisli}, 
assume w.l.o.g.\ that $\BC=\BV_{\BBT}$ for $T=U(J(\argument),\iobj)$.
Let $\IB{f}{g} = U(Jf,Jg)$. \autoref{lem:cof} yields the desired 
guarded morphisms $\eps_{X,Y}\c\IB{X}{Y}\to \IB{(X+Y)}{\iobj}$, which satisfy
the requisite universal property, obtained by interpreting the corresponding property
of \autoref{lem:cof} in~$\BC=\BV_{\BBT}$.

($\Leftarrow$) Conversely, given a bifunctor $\IB{}{}\c\BV\times\BV\to\BV$ with
the described properties, let~$\BBT$ be the monad on $\IB{\argument}{\iobj}$ and apply 
\autoref{lem:cof} with $\BC=\BV_{\BBT}$, $J$ being the free functor $J\c\BV\to\BV_{\BBT}$
and $U(X,Y)=\IB{X}{Y}$.
\end{proof}
\autoref{thm:gpm} provides a bijective correspondence between morphisms 
$f\c X\to Y\gsep Z$ in~$\BC$ and the morphisms $f^\sharp\c X\to Y\pm Z$ in $\BV$, representing 
them. Uniqueness of the $f^\sharp$ is easily seen to 
be equivalent to the monicity of the $\eps_{X,Z}$. 

\section{Guarded Parameterized Monads}\label{sec:gpm}
\autoref{thm:gpm} describes guardedness as a 
certain bifunctor $\IB{}{}\c\BV\times\BV\to\BV$ and a family of morphisms 
$(\eps_{X,Y}\c \IB{X}{Y}\to X\gsep Y)_{X,Y\in |\BV|}$, so that the 
guardedness predicate is derivable.
However, the guardedness laws are still 
formulated in terms of this predicate, and not in terms of
$\IB{}{}$ and $\eps$.
To make the new definition of guardedness self-contained, we must 
identify a collection of canonical morphisms and a complete set of equations 
relating them, in the sense that the guardedness laws for all derived
guarded morphisms follow. For example, by applying~$(\argument)^\sharp$
to the composition
\begin{align*}
X\pm (Y + Z)  \xto{\eps_{X,Y+Z}} (X+(Y+Z))\pm\iobj\iso ((X+Y)+Z))\pm\iobj
\end{align*}
we obtain a morphism $\ups_{X,Y,Z}\c X\pm (Y + Z) \to (X+Y)\pm Z$, which represents
weakening of the guardedness guarantee: in $X\pm (Y + Z)$ the guarded part is~$Y + Z$,
while in $(X+Y)\pm Z$ the guarded part is only $Z$. It should not make a difference
though if starting from $X\pm (Y + (Z + V))$ we apply $\ups$ twice or rearrange 
$Y + (Z + V)$ by associativity and subsequently apply $\ups$ only once -- the results must 
be canonically isomorphic, which is indeed provable. Similarly to this case we 
introduce further morphisms and derive laws relating them. We then prove
that the resulting axiomatization enjoys a coherence property (\autoref{thm:coherence}) 
in the style of Mac Lane's coherence theorem for (symmetric) monoidal 
categories~\cite{Mac-Lane71}. In what follows, we switch from coproducts to 
general symmetric tensor products, as coherence can only hold if the corresponding 
structure is not involved. 
\begin{definition}[Guarded Parameterized Monad]\label{def:pm}
A \emph{guarded parameterized monad} on a symmetric monoidal category $(\BV,\ten,I)$
consists of a bifunctor $\IB{}{}\c\BV\times\BV\to\BV$ and natural transformations 
\begin{flalign*}
  &&\eta\c&\; A \to A \pm I,&&\\* 
  &&\xi\c&\;  (A\pm B) \pm C \to A\pm (B\ten C),  	&\ups\c&\; A\pm (B \ten C)  \to (A\ten B)\pm C,&&\\* 
  &&\zeta\c&\;  A\pm (B\pm C) \to A\pm (B \ten C),    	&\chi\c&\; A\pm B \ten C\pm D\to (A \ten C)\pm (B \ten D) .&&%
\end{flalign*}
such that the following diagrams commute, where $\iso$ refers to the obvious canonical 
isomorphisms 
{\allowdisplaybreaks %
\begin{eqnarray*}
\begin{tikzcd}[column sep=1em, row sep=2ex]
(A\pm I)\pm B \ar[rr,"\xi"] && A\pm (I\ten B)\\                
&A\pm B
	\ar[ul,"\eta\pm\id"]
	\ar[ur, "\iso"'] 
\\
(A\pm B)\pm I 
	\ar[rr,"\xi"] && 
A\pm (B\ten I)\\                
&A\pm B
	\ar[ul,"\eta"]
	\ar[ur, "\iso"'] 
\end{tikzcd}
\quad~
\begin{tikzcd}[column sep=2em, row sep=normal]
((A\pm B)\pm C)\pm D
	\rar["\xi"]
	\dar["\xi\pm\id"'] & 
(A\pm B)\pm (C\ten D)
	\ar[dd,"\xi"]
\\
(A\pm (B\ten C))\pm D
	\dar["\xi"'] & 
\\
A\pm ((B\ten C)\ten D)
	\rar["\iso"] & 
A\pm (B\ten(C\ten D))
\end{tikzcd}
\end{eqnarray*}

\begin{eqnarray*}
\begin{tikzcd}[column sep=.5em, row sep=2ex]
A\pm (I \ten B) 
	\ar[rr,"\ups"] & & 
(A\ten I)\pm B
 \\
&  
A\pm B\ar[ur,"\iso"']\ar[ul,"\iso"]  &   
\end{tikzcd}
\quad
\begin{tikzcd}[column sep=1.5em, row sep=normal]
  A\pm (B\ten (C\ten D))
  	\rar["\iso"]
  	\dar["\ups"']& 
  A\pm ((B \ten C)\ten D)
  	\ar[dd,"\ups"]
  \\
  (A\ten B)\pm (C\ten D)
  	\dar["\ups"']& 
  \\
  ((A\ten B)\ten C)\pm D
  	\rar["\iso"]& 
  (A\ten (B\ten C))\pm D  
\end{tikzcd}
\hspace{2em} 
\end{eqnarray*}
\begin{eqnarray*}
\begin{tikzcd}[column sep=1.5em, row sep=normal]
  A \ten B 
	\dar["\eta"']
  	\rar["\eta\ten\eta"]& 
  A\pm I\ten B\pm I
  	\dar["\chi"]\\
  (A\ten B)\pm I
  	\rar["\iso"]& 
  (A\ten B)\pm (I\ten I)   
\end{tikzcd}
\quad
\begin{tikzcd}[column sep=.5em, row sep=normal]
  A\pm B\ten C\pm D 
  	\dar["\chi"']
  	\rar["\iso"]&[1em] 
  C\pm D\ten A\pm B
  	\dar["\chi"]\\
  (A\ten C)\pm (B \ten D)
  	\rar["\iso"] & 
  (C\ten A)\pm (D \ten B)   
\end{tikzcd}\quad
\end{eqnarray*}
\begin{eqnarray*}
\qquad\quad\begin{tikzcd}[column sep=2em, row sep=normal]
  A\pm B\ten (C \pm D\ten E \pm F)
  	\dar["\id\pm\chi"']
  	\rar["\iso"] &
  (A\pm B\ten C \pm D)\ten E \pm F
  	\dar["\chi\pm\id"]\\
  A\pm B\ten (C \ten E)\pm (D \ten F)
  	\dar["\chi"'] & 
  (A\ten C)\pm (B \ten D)\ten E \pm F
  	\dar["\chi"]\\
  (A\ten (C\ten E))\pm (B\ten (D\ten F))
  	\rar["\iso"] & 
  ((A\ten C)\ten E)\pm ((B\ten D)\ten F) 
\end{tikzcd}\qquad\\[-1ex]
\end{eqnarray*}
\begin{eqnarray*}
\begin{tikzcd}[column sep=-1em, row sep=2ex]
(A\pm B)\pm (I\pm I) 
	\ar[rr,"\chi"] && 
(A\ten I)\pm (B\ten I)\\                
&
(A\pm B)\ten I
	\ar[ul,"\id\ten\eta"]
	\dar[ur, iso] 
\end{tikzcd}
\quad
\begin{tikzcd}[column sep=1em, row sep=2ex]
A\pm (B\pm I) \ar[rr,"\zeta"] && A\pm (B\ten I)\\                
&A\pm B\ar[ul,"\id\pm\eta"]\dar[ur, iso] 
\end{tikzcd}\qquad\\[-4ex]
\end{eqnarray*}
\begin{eqnarray*}
\begin{tikzcd}[column sep=8em, row sep=normal]
&[-7em] A\pm ((B\pm C)\pm (D\pm E))
  \ar[dl,"\zeta"']\ar[dr,"\id\pm\zeta"] &[-7em]\\[-4ex]
A\pm ((B\pm C)\ten (D\pm E))\dar["\id\pm\chi"']&   & A\pm ((B\pm C)\pm (D\ten E))\dar["\id\pm\xi"]\\
A\pm ((B\ten D)\pm (C\ten E))\dar["\zeta"']&  & A\pm (B\pm(C\ten (D \ten E)))\dar["\zeta"] \\
A\pm ((B\ten D)\ten (C\ten E))\ar[rr, "\cong"] &  & A\pm (B\ten (C\ten (D \ten E)))
\end{tikzcd}
\end{eqnarray*}
\begin{eqnarray*}
\begin{tikzcd}[column sep=8em, row sep=normal]
&[-7em] (A\pm (B\pm C))\pm (D\pm E) 
  \ar[dl,"\xi"']\ar[dr,"\zeta\pm\id"] &[-7em]\\[-4ex]
A\pm ((B\pm C)\ten (D\pm E))\dar["\id\pm\chi"']&& 	(A\pm (B\ten C))\pm (D\pm E)\dar["\zeta"]\\
A\pm ((B\ten D)\pm (C\ten E))\dar["\zeta"']&&		(A\pm (B\ten C))\pm (D\ten E)\dar["\xi"] \\
A\pm ((B\ten D)\ten (C\ten E))\ar[rr, "\cong"]&&	A\pm ((B\ten C)\ten (D \ten E))
\end{tikzcd}
\end{eqnarray*}
\begin{eqnarray*}
&\begin{tikzcd}[column sep=2em, row sep=normal]
((A\pm B) \pm C) \ten ((D\pm E) \pm F)
  \rar["\chi"]
  \dar["\xi\ten\xi"']&
(A\pm B\ten D\pm E) \pm (C\ten F)
  \dar["\chi\pm\id"] \\
A\pm (B \ten C)\ten D\pm (E\ten F)
  \dar["\chi"'] & 
((A\ten D)\pm (B\ten E)) \pm (C\ten F) 
  \dar["\xi"]\\
(A\ten D)\pm ((B \ten C)\ten (E\ten F))
  \rar["\iso"] & 
(A\ten D)\pm ((B\ten E)\ten (C\ten F))
\end{tikzcd}
\end{eqnarray*}
\begin{eqnarray*}
&\begin{tikzcd}[column sep=2em, row sep=normal]
(A\pm (B\pm C))\ten (D\pm (E\pm F))
  \dar["\zeta\ten\zeta"']
  \rar["\chi"]&
(A\ten D)\pm (B\pm C \ten E\ten F)
  \dar["\id\pm\chi"] \\
A\pm (B \ten C)\ten D\pm (E\ten F)
  \dar["\chi"'] & 
(A\ten D)\pm ((B\ten E)\pm (C\ten F)) 
  \dar["\zeta"]\\
(A\ten D)\pm ((B \ten C)\ten (E\ten F))
  \rar["\iso"] & 
(A\ten D)\pm ((B\ten E)\ten (C\ten F))
\end{tikzcd}
\end{eqnarray*}
\begin{eqnarray*}
&\hspace{0em}\begin{tikzcd}[column sep=2em, row sep=normal]
(A\pm B)\pm (C\pm D \ten E\pm F) 
  \rar["\ups"]
  \dar["\id\pm\chi"'] &
(A\pm B\ten C\pm D) \pm (E\pm F)
  \dar["\chi\pm\id"]
\\
(A\pm B)\pm ((C\ten E)\pm (D\ten F)) 
  \dar["\zeta"'] & 
((A\ten C) \pm (B\ten D)) \pm (E\pm F)
  \dar["\zeta"]
\\
(A\pm B)\pm ((C\ten E)\ten (D\ten F)) 
  \dar["\xi"'] & 
((A\ten C) \pm (B\ten D))\ten (E\pm F)
  \ar[dd,"\xi"]
\\
A\pm (B\ten ((C\ten D) \ten (E\ten F))) 
  \dar["\iso"'] & 
  \\
A \pm (C\ten ((B\ten D)\ten (E\ten F)))
  \rar["\ups"] & 
(A\ten C) \pm ((B\ten D)\ten (E\ten F))
\end{tikzcd}
\end{eqnarray*}
\begin{eqnarray*}
\begin{tikzcd}[column sep=3em, row sep=normal]
A\pm B\ten C\pm (D \ten E) 
  \ar[r,"\id\ten\ups"]
  \dar["\chi"'] &[-1em]  
A\pm B\ten (C\ten D) \pm E 
  \dar["\chi"]
\\
(A\ten C)\pm (B\ten (D \ten E))
  \dar["\iso"']  & 
(A\ten (C\ten D))\pm (B\ten E) 
  \dar["\iso"]
\\
(A\ten C)\pm (D\ten (B\ten E))
  \rar["\ups"] &
((A\ten C)\ten D))\pm (B\ten E)
\end{tikzcd}
\\[2ex]
\begin{tikzcd}[column sep=2em, row sep=normal]
(A\pm (B \ten C)) \pm D
  \rar["\ups\pm\id"]
  \dar["\xi"'] 
  & 
((A\ten B) \pm C) \pm D
  \ar[dd, "\xi"]
  \\
A\pm ((B\ten C)\ten D)
  \dar["\iso"'] 
&
\\
A\pm (B\ten (C\ten D)) 
  \rar["\ups"] 
&
(A\ten B)\pm (C\ten D) 
\end{tikzcd}\qquad\quad
\end{eqnarray*}
\begin{eqnarray*}
&\hspace{-1em}
\begin{tikzcd}[column sep=2em, row sep=normal]
A\pm (B\pm (C \ten D)) 
  \ar[r,"\id\pm\ups"]
  \dar["\zeta"'] & 
A\pm ((B\ten C) \pm D) 
  \dar["\zeta"]
\\
A\pm (B\ten (C\ten D))
  \rar["\iso"] & 
A\pm ((B\ten C)\ten D)
\end{tikzcd}
\end{eqnarray*}
} %
\end{definition}
\begin{remark}\label{rem:grad}
The first three laws (relating $\eta$ and $\xi$) identify guarded parameterized 
monads as \emph{parametric monads} in the sense of Melli{\'e}s~\cite{Mellies17}, and subsequently
renamed to \emph{graded monads}~\cite{FujiiKatsumataEtAl16}. In our case, more specifically, 
$\pm$ is a $\BV$-graded monad on $\BV$.
\end{remark}
The relevance of the presented axiomatization is certified by the following
\begin{theorem}[Coherence]\label{thm:coherence}
Let $\CE_1$, $\CE_2$ and $\CE_2'$ be expressions, built from~$\ten$, $\pm$ and~$I$ over 
a set of letters, in such a way that $\CE_1$ and $\CE_2\pm \CE_2'$ contain every letter at 
most once and neither~$\CE_2$ nor $\CE_2'$ contain $\pm$. Let $f$ and $g$ be two expressions
built with $\ten$ and $\pm$ over identities, $\eta$, $\ups$, 
$\xi$, $\zeta$, $\chi$, associators, unitors, braidings and inverses of associators and 
unitors, in such a way that the judgements $f\c\CE_1\to\CE_2\pm\CE_2'$
and $g\c\CE_1\to\CE_2\pm\CE_2'$ are formally valid. Then $f=g$ follows from the axioms of guarded parameterized monads. 
\end{theorem}
\begin{proof}
For the sake of the present proof, let us introduce some nomenclature. We will use 
the following names correspondingly for associators, right unitors and braidings:
\begin{gather*}
\alpha\c A\ten (B\ten C)\to (A\ten B)\ten C,\qquad
\rho\c A\ten I\to A,\qquad
\gamma\c A\ten B\to B\ten A.
\end{gather*}
We dispense with the left unitor, since our monoidal structure is symmetric. 

Let us refer to the expressions built from~$\ten$, $\pm$ and~$I$ over some alphabet 
of \emph{object names}, fixed globally from now on, in such a way that every 
object name occurs at most once, as \emph{object expressions}.
We refer to the expressions built with~$\ten$ and $\pm$ over $\id$, $\eta$, $\ups$, $\xi$, $\zeta$, $\chi$,
$\alpha$, $\alpha^\mone$, $\rho$, $\rho^\mone$, $\gamma$ as \emph{morphism expressions}. 
An \emph{isomorphism expression} is then a morphism expression that does not involve 
$\eta$, $\ups$, $\xi$, $\zeta$, $\chi$.
Every morphism expression $f$ unambiguously identifies object expressions $\CE_1$
and $\CE_2$ for which the judgement $f\c\CE_1\to\CE_2$ is formally valid (that is,
in any category, where we can interpret $f$, $\CE_1$ and $\CE_2$, $f$ is a morphism 
from $\CE_1$ to~$\CE_2$).
For two morphism expressions $f,g\c\CE_1\to\CE_2$, let $f\equiv g$ denote \emph{`$f=g$ follows from the axioms of guarded parameterized 
monads`}. An object expression is \emph{normal} if it is of the form $\CE\pm \CE'$ and $\CE$ and $\CE'$ do not contain $\pm$. 
A morphism expression is \emph{simple} if it is a composition of isomorphism expressions 
between normal form expressions and instances of~$\ups$.

For an object expression~$\CE$, we define object expressions $\nf_1(\CE)$ and $\nf_2(\CE)$ recursively with the clauses:
\begin{itemize}[wide]
  \item $\nf_1(\CE) = \CE$, $\nf_2(\CE) = I$ if $\CE=I$ or $\CE$ is an object name;
  \item $\nf_1(\CE \ten \CE') = \nf_1(\CE) \ten \nf_1(\CE')$, $\nf_2(\CE \ten \CE') = \nf_2(\CE) \ten \nf_2(\CE')$; 
  \item $\nf_1(\CE \pm \CE') = \nf_1(\CE)$, $\nf_2(\CE \pm \CE') = \nf_2(\CE) \ten (\nf_1(\CE') \ten \nf_2(\CE'))$. 
\end{itemize} 
Let $\nf(\CE) = \nf_1(\CE)\pm\nf_2(\CE)$, so $\nf(\CE)$ is normal. For any object expression $\CE$ we also define
a \emph{normalization morphism expression} $\nm(\CE)\c \CE\to\nf(\CE)$, by induction as follows:
\begin{itemize}[wide]
  \item $\nm(\CE) = \eta$ if $\CE=I$ or $\CE$ is an object name;
  \item $\nm(\CE \ten \CE') = \chi\comp (\nm(\CE) \ten\nm(\CE'))$; 
  \item $\nm(\CE \pm \CE') = \xi\comp\zeta\comp (\nm(\CE) \pm\nm(\CE'))$. 
\end{itemize} 
The statement of the theorem will follow from the following subgoals.
\begin{enumerate}[wide]
  \item\label{it:coh1} If a morphism expression $f\c \CE\to \CE'$ does not contain $\ups$
  then $\nm(\CE') \comp f \equiv \nm(\CE)\comp g$ for some isomorphism expression $g$.
  \item\label{it:coh2} If a morphism expression $f\c \CE\to \CE'$ does not contain $\eta$, 
  $\xi$, $\zeta$ and $\chi$, then there exists a simple morphism expression $g\c \nf(\CE)\to \nf(\CE')$, 
  such that $\nm(\CE')\comp f \equiv g\comp\nm(\CE)$.
  \item\label{it:coh3} For any two simple morphism expressions $f\c\CE\to\CE'$ and $g\c\CE\to\CE'$, 
  $f\equiv g$.
  \item\label{it:coh4} For every normal object expression $\CE$, $\nm(\CE)\c\CE\to\nf(\CE)$
  is an isomorphism expression.
\end{enumerate}
Indeed, given $f,g\c \CE\to \CE'$ with normal $\CE'$, to prove $f\equiv g$, it suffices 
to prove that $f$ is equal to $\CE\xto{\nm(\CE)}\nf(\CE)\xto{f'} \CE'$ for some simple 
$f'$  -- the analogous statement would be true for~$g$, and we would be done by~\eqref{it:coh3}.
In order to construct $f'$, let us represent $f$ as a composition $f_n\comp\ldots\comp f_1$ where every~$f_i$ with even 
$i$ contains precisely one occurrence of $\ups$ and every $f_i$ with odd $i$ contains 
no occurrences of $\ups$. We obtain 
\begin{equation*}
\begin{tikzcd}[column sep=4em, row sep=normal]
\CE
  \dar["\nm(\CE)"']\rar["f_1"] & 
\CE_1
  \dar["\nm(\CE_1)"']\rar["f_2"] & 
\CE_2\dar["\nm(\CE_2)"']\rar[phantom,"\dots"] & \CE_n\dar["\nm(\CE_n)"]\\
\nf(\CE)
  \rar["\iso"] & 
\nf(\CE_1)\rar["f_2'"] & 
\nf(\CE_2)\rar[phantom,"\dots"] & \nf(\CE_n) 
\end{tikzcd}
\end{equation*}
where $\CE=\CE_n$, every odd diagram commutes by~\eqref{it:coh1} and every even diagram commutes 
by~\eqref{it:coh2}. Note that $\nm(\CE_n)$ is an isomorphism expression 
by~\eqref{it:coh1}, and therefore we obtain the desired presentation for $f$,
by composing the left vertical arrow, the bottom horizontal sequence of arrows
and the inverse of the right vertical arrow. 
It remains to show the subgoals \eqref{it:coh1}--\eqref{it:coh4}. 
\begin{enumerate}[wide]
  \item[\eqref{it:coh1}] 
  We strengthen the claim by demanding the requisite isomorphism $g$ to be of 
  the form $g_1\pm g_2$ and proceed by structural induction on $f$.

  \emph{Induction Base:} $f\in \{\id,\eta\comma \xi\comma\zeta\comma \chi\comma \rho\comma\rho^\mone\comma\alpha\comma\alpha^\mone\comma\gamma\}$.
  If $f=\id$, we are done trivially by taking $g=f$. Consider $f=\eta\c\CE\to\CE\pm I$.
  Then the following diagram commutes, and we obtain the requisite isomorphism 
  $g$ as the bottom horizontal morphism:
    {\small
  \begin{equation*}
  \begin{tikzcd}[column sep=2em, row sep=normal]
  \CE\ar[rr,"\eta"]\ar[ddd,"\nm(\CE)"'] &&  \CE\pm I\dar["\nm(\CE) \pm\eta"]\ar[ddd, shiftarr = {xshift=82}, "\nm(\CE\pm I)"] \\
   &  &  (\nf_1(\CE)\pm\nf_2(\CE))\pm (I\pm I)\dar["\zeta"]\\
   & (\nf_1(\CE)\pm\nf_2(\CE))\pm I\ar[ur,"\id\pm\eta"]\ar[r,iso]\dar["\xi"] &  (\nf_1(\CE)\pm\nf_2(\CE))\pm (I\ten I)\dar["\xi"]\\
  \nf_1(\CE)\pm\nf_2(\CE)\ar[ur,"\eta"]\rar["\iso"] &\nf_1(\CE)\pm(\nf_2(\CE)\ten I) \ar[r, iso]&  \nf_1(\CE)\pm (\nf_2(\CE) \ten (I \ten I))                            
  \end{tikzcd}
  \end{equation*}}

  \noindent The remaining cases are handled by producing analogous commutative diagrams, 
  which are given below. We do not treat $f=\rho^\mone$ and $f=\alpha^\mone$, as these 
  cases are obtained from the corresponding cases $f=\rho$ and $f=\alpha$ by 
  flipping the corresponding diagrams.
  To save space and maintain readability, we write $\CA,\CB,\CC,\CD$ for object 
  expressions, $\CA_1,\CB_1,\CC_1,\CD_1$ for $\nf_1(\CA),\nf_1(\CB),\nf_1(\CC),\nf_1(\CD)$
  and $\CA_2,\CB_2,\CC_2,\CD_2$ for $\nf_2(\CA),\nf_2(\CB),\nf_2(\CC),\nf_2(\CD)$  
  respectively.
  {\small
  \begin{equation*}
  \begin{tikzcd}[column sep=-6em, row sep=normal]
  (\CA\pm\CB) \pm \CC
  	\dar["(\nm\pm\nm)\pm\nm"'] 
  	\ar[rr,"\xi"]\ar[d, ] && 
  \CA\pm (\CB\ten\CC)
    	\dar["\nm\pm (\nm\ten\nm)"]  
  \\
  ((\CA_1\pm\CA_2)\pm(\CB_1\pm\CB_2)) \pm (\CC_1\pm\CC_2)
  	\dar["\zeta\pm\id"'] 
  	\ar[rr,"\xi" ] &  &
  (\CA_1\pm\CA_2)\pm ((\CB_1\pm\CB_2)\ten(\CC_1\pm\CC_2))
  	\dar["\id\pm\chi"]  
  \\
  ((\CA_1\pm\CA_2)\pm(\CB_1\ten\CB_2)) \pm (\CC_1\pm\CC_2) 
  	\ar[dr,"\zeta"]
  	\ar[dd,"\xi\pm\id"']&& 
  (\CA_1\pm\CA_2)\pm ((\CB_1\ten\CC_1)\pm(\CB_2\ten\CC_2))
  	\ar[ddd,"\zeta"] 
  \\
   &((\CA_1\pm\CA_2)\pm(\CB_1\ten\CB_2)) \pm (\CC_1\ten\CC_2)
  	\ar[ddd,"\xi"]
  	\ar[ddl, "\xi\pm\id"', bend left=25] & 
  \\[-2ex]
  (\CA_1\pm (\CA_2\ten (\CB_1\ten\CB_2))) \pm (\CC_1\pm\CC_2)  
  	\ar[d, "\zeta"']&  	& 
  \\
  (\CA_1\pm (\CA_2\ten (\CB_1\ten\CB_2))) \pm (\CC_1\ten\CC_2)
  \ar[ddd, "\xi"'] &
  	& 
  (\CA_1\pm\CA_2)\pm ((\CB_1\ten\CC_1)\ten(\CB_2\ten\CC_2))
  	\ar[ddd,"\xi"]
  \\[-1ex]
 	 &(\CA_1\pm\CA_2)\pm((\CB_1\ten\CB_2) \ten (\CC_1\ten\CC_2))
  	\ar[ur, "\cong"]
  	\ar[d, "\xi"] & 
  \\
  &
  \CA_1\pm(\CA_2\ten((\CB_1\ten\CB_2) \ten (\CC_1\ten\CC_2)))
  	\ar[dr, "\cong"]
  	\ar[dl, "\cong"', bend left=10]&  
  \\[-1ex]
  \CA_1\pm (((\CA_2\ten (\CB_1\ten\CB_2))) \ten (\CC_1\ten\CC_2)) 
  	\ar[rr,"\cong"]	  && 
  \CA_1\pm(\CA_2\ten ((\CB_1\ten\CC_1)\ten(\CB_2\ten\CC_2))) 
  \end{tikzcd}
  \end{equation*}
}

{\small
  \begin{equation*}
  \begin{tikzcd}[column sep=normal, row sep=3.75ex]
  \CA\pm (\CB\pm\CC) 
  	\rar["\zeta"]	
  	\dar["\nm\pm(\nm\pm\nm)"']& 
  \CA\pm (\CB\ten\CC)
  	\dar["\nm\pm(\nm\ten\nm)"]    
  \\       
  (\CA_1\pm\CA_2)\pm ((\CB_1\pm\CB_2)\pm(\CC_1\pm\CC_2)) 
  	\rar["\zeta"]
  	\dar["\id\pm\zeta"']& 
  (\CA_1\pm\CA_2)\pm ((\CB_1\pm\CB_2)\ten(\CC_1\pm\CC_2))      
  	\dar["\id\pm\chi"]     
  \\ 
  (\CA_1\pm\CA_2)\pm ((\CB_1\pm\CB_2)\pm(\CC_1\ten\CC_2)) 
  	\dar["\id\pm\xi"']& 
  (\CA_1\pm\CA_2)\pm ((\CB_1\ten\CC_1)\pm(\CB_2\ten\CC_2))
  	\ar[dd,"\zeta"]           
  \\       
  (\CA_1\pm\CA_2)\pm (\CB_1\pm(\CB_2\ten(\CC_1\ten\CC_2))) 
  	\dar["\zeta"']&          
  \\       
  (\CA_1\pm\CA_2)\pm (\CB_1\ten(\CB_2\ten(\CC_1\ten\CC_2))) 
  	\dar["\xi"']
  	\rar["\iso"]&   
  (\CA_1\pm\CA_2)\pm ((\CB_1\ten\CC_1)\ten(\CB_2\ten\CC_2)) 
  	\dar["\xi"]     
  \\       
  \CA_1\ten(\CA_2\pm (\CB_1\ten(\CB_2\ten(\CC_1\ten\CC_2)))) 
  	\rar["\iso"]&   
  \CA_1\pm(\CA_2\ten ((\CB_1\ten\CC_1)\ten(\CB_2\ten\CC_2)))     
  \end{tikzcd}
  \end{equation*}
}

{\small
  \begin{equation*}
  \begin{tikzcd}[column sep=-9em, row sep=2.65ex]
  \CA\pm\CB \ten \CC\pm \CD
  	\ar[rr,"\chi"]
  	\dar["(\nm\pm\nm)\ten(\nm\pm\nm)"']&[-.5em] &[1.5em] 
  (\CA \ten \CC)\pm (\CB \ten \CD)
  	\dar["(\nm\ten\nm)\pm(\nm\ten\nm)"]
  \\[2ex]
  \parbox{4cm}{\centering $((\CA_1\pm\CA_2)\pm(\CB_1\pm\CB_2)) \ten ((\CC_1\pm\CC_2)\pm (\CD_1\pm\CD_2))$}
  	\ar[dddd,"\zeta\ten\zeta"']
  	\ar[rr,"\chi"]& &
  \parbox{4cm}{\centering $((\CA_1\pm\CA_2) \ten (\CC_1\pm\CC_2))\pm ((\CB_1\pm\CB_2) \ten (\CD_1\pm\CD_2))$}
  	\ar[dd,"\chi\pm\chi"]    
  	\ar[dl, "\id\pm\chi"']  
  \\
  & ((\CA_1\pm\CA_2) \ten (\CC_1\pm\CC_2))\pm ((\CB_1\ten\CD_1)\pm(\CB_2\ten\CD_2))
  \ar[dd, "\zeta"] & 
  \\
  &  &   
  \parbox{4cm}{\centering $((\CA_1\ten \CC_1)\pm(\CA_2\ten\CC_2))\pm ((\CB_1\ten\CD_1)\pm(\CB_2\ten\CD_2))$} 
  	\ar[dd,"\zeta"]
  \\
  &((\CA_1\pm\CA_2) \ten (\CC_1\pm\CC_2))\pm ((\CB_1\ten\CD_1)\ten(\CB_2\ten\CD_2))
  	\ar[dr, "\chi\pm\id"']
  	\ar[dd, "\cong"'] &
  \\
  \parbox{4cm}{\centering $((\CA_1\pm\CA_2)\pm(\CB_1\ten\CB_2)) \ten ((\CC_1\pm\CC_2)\pm (\CD_1\ten\CD_2))$}
  	\ar[ddd,"\xi\ten\xi"']
  	\ar[dr,"\chi", bend right=8] & & 
  \parbox{4cm}{\centering $((\CA_1\ten \CC_1)\pm(\CA_2\ten\CC_2))\pm ((\CB_1\ten\CD_1)\ten(\CB_2\ten\CD_2))$}
  	\ar[dddd, "\xi"]
  	\ar[ddl, "\cong"', bend left=26]
  \\
  & ((\CA_1\pm\CA_2)\ten (\CC_1\pm\CC_2)) \pm ((\CB_1\ten\CB_2)\ten (\CD_1\ten\CD_2))
    \dar["\chi\ten\id"] &    
  \\[5ex]
    & ((\CA_1\ten\CC_1)\pm (\CA_2\ten\CC_2)) \pm ((\CB_1\ten\CB_2)\ten (\CD_1\ten\CD_2))
    \ar[ddd, "\xi"] &
  \\
  \parbox{4cm}{\centering $(\CA_1\pm(\CA_2\ten(\CB_1\ten\CB_2))) \ten (\CC_1\pm(\CC_2\ten (\CD_1\ten\CD_2)))$}
  	\dar["\chi"']& &
  \\[1ex]
  \parbox{6cm}{\centering $(\CA_1\ten\CC_1)\pm ((\CA_2\ten(\CB_1\ten\CB_2)) \ten (\CC_2\ten (\CD_1\ten\CD_2)))$}
  	\ar[rd,"\cong", bend right=8]& &
  \parbox{6cm}{\centering $(\CA_1\ten \CC_1)\pm ((\CA_2\ten\CC_2)\ten ({(\CB_1\ten\CD_1)}\ten(\CB_2\ten\CD_2)))$}
  	\ar[dl, "\cong"', bend left=7]
  \\
  & (\CA_1\ten\CC_1)\pm ((\CA_2\ten\CC_2) \ten ((\CB_1\ten\CB_2)\ten (\CD_1\ten\CD_2))) &
  \end{tikzcd}
  \end{equation*}
}

{\small
  \begin{equation*}
  \begin{tikzcd}[column sep=2em, row sep=6ex]
  \CA\ten I
  	\dar["\nm\ten\id"']  
  	\ar[rr, "\rho"]
  	&[2em] & 
  \CA\dar["\nm"]
  \\
  (\CA_1\pm\CA_2)\ten I 
  	\ar[rr,"\rho"]
  	\ar[d,"\id\ten\eta"']&& 
  \CA_1\pm\CA_2
  	\ar[dd, equals]
  \\ 
  (\CA_1\pm\CA_2)\ten (I\pm I)
  	\dar["\chi"'] & & 
  \\
  (\CA_1\ten I)\pm (\CA_2\ten I)
  	\ar[rr, "\iso"]& 
& 
  \CA_1\pm\CA_2
  \end{tikzcd}
  \end{equation*}
}

{\small
\begin{equation*}
\begin{tikzcd}[column sep=1em, row sep=normal]
  \CA\ten (\CB\ten\CC)
  	\ar[rr,"\alpha"]
  	\dar["\nm\tensor(\nm\ten\nm)"'] & &[3em] 
  (\CA\ten \CB)\ten \CC
  	\dar["(\nm\ten\nm)\ten\nm"]
\\
  (\CA_1\pm\CA_2)\ten ((\CB_1\pm\CB_2)\ten(\CC_1\pm\CC_2))
  	\dar["\id\ten\chi"']
  	\ar[rr,"\alpha"]
  	& &  
  ((\CA_1\pm\CA_2)\ten (\CB_1\pm\CB_2))\ten(\CC_1\pm\CC_2)
  	\dar["\chi\ten\id"]
\\
  (\CA_1\pm\CA_2)\ten ((\CB_1\ten\CC_1)\pm(\CB_2\ten\CC_2))
  	\dar["\chi"']& &
  ((\CA_1\ten\CB_1)\pm (\CA_2\ten\CB_2))\ten(\CC_1\pm\CC_2)
  	\dar["\chi"]
\\
  (\CA_1\ten(\CB_1\ten\CC_1)) \pm(\CA_2\ten(\CB_2\ten\CC_2))
  	\ar[rr, "\cong"]& &
  ((\CA_1\ten\CB_1)\ten\CC_1) \pm ((\CA_2\ten\CB_2)\ten\CC_2)
  \end{tikzcd}
  \end{equation*}
}

{\small
  \begin{equation*}
  \begin{tikzcd}[column sep=normal, row sep=normal]
  \CA\ten\CB
  	\rar["\gamma"]
  	\dar["\nm\ten\nm"']& 
  \CB\ten\CA
  	\dar["\nm\ten\nm"]\\
  (\CA_1\pm\CA_2)\ten(\CB_1\pm\CB_2)
  	\rar["\gamma"]
  	\dar["\chi"']& 
  (\CB_1\pm\CB_2)\ten(\CA_1\pm\CA_2)
  	\dar["\chi"]\\
  (\CA_1\ten\CB_1)\ten(\CA_2\ten\CB_2)
  	\rar["\iso"]& 
  (\CB_1\ten\CA_1)\ten(\CB_2\ten\CA_2) 
  \end{tikzcd}
  \end{equation*}
}

\noindent
\emph{Induction Step.} Consider $f=f_1\pm f_2\c\CE_1\pm\CE_2\to\CE'_1\pm\CE'_2$.
We obtain the requisite isomorphism from the diagram:
{\small
  \begin{equation*}
  \begin{tikzcd}[column sep=5em, row sep=normal]
  \CE_1\pm\CE_2
  	\rar["f_1\pm f_2"]
  	\dar["\nm(\CE_1)\pm\nm(\CE_2)"']& 
  \CE_1'\pm\CE_2'
  	\dar["\nm(\CE_1')\pm\nm(\CE_2')"]\\
  (\nf_1(\CE_1)\pm\nf_2(\CE_1))\pm (\nf_1(\CE_2)\pm\nf_2(\CE_2))
  	\rar["(e_1\pm u_1)\pm(e_2\pm u_2)"]
  	\dar["\zeta"']& 
  (\nf_1(\CE_1')\pm\nf_2(\CE_1'))\pm (\nf_1(\CE_2')\pm\nf_2(\CE_2'))
  	\dar["\zeta"]\\
  (\nf_1(\CE_1)\pm\nf_2(\CE_1))\pm (\nf_1(\CE_2)\ten\nf_2(\CE_2))
  	\dar["\xi"']& 
  (\nf_1(\CE_1')\pm\nf_2(\CE_1'))\pm (\nf_1(\CE_2')\ten\nf_2(\CE_2'))
  	\dar["\xi"]
  \\
  \nf_1(\CE_1)\ten(\nf_2(\CE_1)\pm (\nf_1(\CE_2)\ten\nf_2(\CE_2)))
  	\rar["\iso"]& 
  \nf_1(\CE_1')\ten(\nf_2(\CE_1')\pm (\nf_1(\CE_2')\ten\nf_2(\CE_2'))) 
  \end{tikzcd}
  \end{equation*}
}

\noindent
Here, the top cell commutes by induction hypothesis, with some isomorphisms 
$e_1$, $e_2$, $u_1$ and~$u_2$, and the bottom cell commutes by naturality of $\zeta$
and $\xi$. The case $f=f_1\ten f_2$ is handled analogously.
  \item[\eqref{it:coh2}] The given morphism expression $f$ can be decomposed as $f_n\comp\ldots\comp f_1$
  in such a way that every $f_i$ contains at most one morphism name (not including $\id$).
  It thus suffices to establish the claim for every $i$. If the involved name is not 
  $\ups$, we are done by the previous clause~\eqref{it:coh1}. We thus continue with the proof 
  of~\eqref{it:coh2} w.l.o.g.\, assuming that $f$ is built from $\id$, $\ups$, $\ten$ and 
  $\pm$, and $\ups$ occurs in $f$ precisely once. 
  
  Consider $f=\ups$. Using the naming conventions from clause~(1), we obtain the following
  commutative diagram:
{\small
  \begin{equation*}
  \begin{tikzcd}[column sep=5em, row sep=normal]
  \CA\pm (\CB\ten\CC) 
  	\rar["\ups"]	
  	\dar["\nm\pm(\nm\ten\nm)"']& 
  (\CA\ten\CB)\pm\CC
  	\dar["(\nm\ten\nm)\pm\nm"]    
  \\       
  (\CA_1\pm\CA_2)\pm ((\CB_1\pm\CB_2)\ten(\CC_1\pm\CC_2)) 
  	\rar["\ups"]
  	\dar["\id\pm\chi"']& 
  ((\CA_1\pm\CA_2)\ten (\CB_1\pm\CB_2))\pm(\CC_1\pm\CC_2)      
  	\dar["\chi\pm\id"]     
  \\ 
  (\CA_1\pm\CA_2)\pm ((\CB_1\ten\CC_1)\pm(\CB_2\ten\CC_2)) 
  	\dar["\zeta"']& 
  ((\CA_1\ten\CB_1)\pm (\CA_2\ten\CB_2))\pm(\CC_1\pm\CC_2)
  	\ar[d,"\zeta"]           
  \\       
  (\CA_1\pm\CA_2)\pm ((\CB_1\ten\CC_1)\ten(\CB_2\ten\CC_2))
  	\dar["\xi"']& 
  ((\CA_1\ten\CB_1)\pm (\CA_2\ten\CB_2))\pm(\CC_1\ten\CC_2)
  	\dar["\xi"]         
  \\       
  \CA_1\pm(\CA_2\ten ((\CB_1\ten\CC_1)\ten(\CB_2\ten\CC_2))) 
  	\rar["\ups\comp(\id\pm e)"]&   
  (\CA_1\pm\CB_1)\pm ((\CA_2\ten\CB_2)\ten(\CC_1\ten\CC_1)) 
  \end{tikzcd}
  \end{equation*}
}
which yields $g=\ups\comp(\id\pm e)$. Notably, $e$ is an isomorphism. Next,
we handle the cases $f=\ups\pm\id$, $f=\id\pm\ups$, $f=\ups\ten\id$ and $f=\id\ten\ups$.
For $f=\ups\pm\id$, we have 
{\small
  \begin{equation*}
  \begin{tikzcd}[column sep=-5.5em, row sep=normal]
  \CE_1\pm\CE_2
  	\ar[rr, "\ups\pm\id"]
  	\dar["\nm(\CE_1)\pm\nm(\CE_2)"']& &
  \CE_1'\pm\CE_2
  	\dar["\nm(\CE_1')\pm\nm(\CE_2)"]
  \\
  (\nf_1(\CE_1)\pm\nf_2(\CE_1))\pm (\nf_1(\CE_2)\pm\nf_2(\CE_2))
  	\ar[rr, "(\ups\comp (\id\pm e))\pm\id"]
  	\dar["\zeta"']& &
  (\nf_1(\CE_1')\pm\nf_2(\CE_1'))\pm (\nf_1(\CE_2)\pm\nf_2(\CE_2))
  	\dar["\zeta"]
  \\
  (\nf_1(\CE_1)\pm\nf_2(\CE_1))\pm (\nf_1(\CE_2)\ten\nf_2(\CE_2))
  	\ar[dd,"\xi"']
  	\ar[dr, "(\id\pm e)\pm\id"']
  	\ar[rr, "(\ups\comp (\id\pm e))\pm\id"]& &
  (\nf_1(\CE_1')\pm\nf_2(\CE_1'))\pm (\nf_1(\CE_2)\ten\nf_2(\CE_2))
  	\ar[dd,"\xi"]
  \\
  & (\nf_1(\CE_1)\pm\CE)\pm (\nf_1(\CE_2)\ten\nf_2(\CE_2))
  	\ar[ru,"\ups\pm\id"']
  	\ar[dd,"\xi"] &
  \\[-2ex]
  \nf_1(\CE_1)\pm(\nf_2(\CE_1)\ten (\nf_1(\CE_2)\ten\nf_2(\CE_2)))
  	\ar[dr, "\id\pm(e\ten\id)"'] & &
  \nf_1(\CE_1')\pm(\nf_2(\CE_1')\ten (\nf_1(\CE_2)\ten\nf_2(\CE_2))) 
  \\
  	& 
  \nf_1(\CE_1)\pm(\CE\ten (\nf_1(\CE_2)\ten\nf_2(\CE_2)))
  	\ar[ur, "\ups\comp e'"'] &
  \end{tikzcd}
  \end{equation*}
}
where $e$ and $e'$ are isomorphisms, and $e$ is the one that we inherit from the 
case $f=\ups$.

For $f=\id\pm\ups$, we have construct the requisite $g$ analogously: 
{\small
  \begin{equation*}
  \begin{tikzcd}[column sep=-5.5em, row sep=normal]
  \CE_1\pm\CE_2
  	\ar[rr, "\id\pm\ups"]
  	\dar["\nm(\CE_1)\pm\nm(\CE_2)"']& &
  \CE_1\pm\CE_2'
  	\dar["\nm(\CE_1)\pm\nm(\CE_2')"]
  \\
  (\nf_1(\CE_1)\pm\nf_2(\CE_1))\pm (\nf_1(\CE_2)\pm\nf_2(\CE_2))
  	\ar[dd,"\zeta"']
  	\ar[dr, "\id\pm(\id\pm e)"']
  	\ar[rr, "\id\pm(\ups\comp (\id\pm e))"]& &
  (\nf_1(\CE_1)\pm\nf_2(\CE_1))\pm (\nf_1(\CE_2')\pm\nf_2(\CE_2'))
  	\ar[dd,"\zeta"]
  \\
  & (\nf_1(\CE_1)\pm\nf_2(\CE_1))\pm (\nf_1(\CE_2)\pm\CE)
  	\ar[ru,"\id\pm\ups"']
  	\ar[dd,"\zeta"] &
  \\[-2ex]
  (\nf_1(\CE_1)\pm\nf_2(\CE_1))\pm (\nf_1(\CE_2)\ten\nf_2(\CE_2))
  	\ar[dr, "\id\pm(\id\ten e)"']
  	\ar[dd, "\xi"'] & &
  (\nf_1(\CE_1)\pm\nf_2(\CE_1))\pm (\nf_1(\CE_2')\ten\nf_2(\CE_2')) 
  	\ar[dd, "\xi"] 
  \\
  & (\nf_1(\CE_1)\pm\nf_2(\CE_1))\pm (\nf_1(\CE_2)\ten\CE)
  	\ar[ur, "\cong"']
  	\ar[dd, "\xi"] &
  \\[-2ex]
  \nf_1(\CE_1)\pm(\nf_2(\CE_1)\ten (\nf_1(\CE_2)\ten\nf_2(\CE_2)))
  	\ar[dr, "\id\pm(\id\ten(\id\ten e))"'] 
  	& &
  (\nf_1(\CE_1)\pm\nf_2(\CE_1))\pm (\nf_1(\CE_2')\ten\nf_2(\CE_2')) 
  \\
  	& 
  \nf_1(\CE_1)\pm(\nf_2(\CE_1)\ten (\nf_1(\CE_2)\ten\CE))
  	\ar[ur, "\cong"'] &
  \end{tikzcd}
  \end{equation*}
}

\noindent
For $f=\ups\ten\id$, we have 
{\small
  \begin{equation*}
  \begin{tikzcd}[column sep=-5.5em, row sep=normal]
  \CE_1\ten\CE_2
  	\ar[rr, "\ups\ten\id"]
  	\dar["\nm(\CE_1)\ten\nm(\CE_2)"']& &
  \CE_1'\ten\CE_2
  	\dar["\nm(\CE_1')\ten\nm(\CE_2)"]
  \\
  (\nf_1(\CE_1)\pm\nf_2(\CE_1))\ten (\nf_1(\CE_2)\pm\nf_2(\CE_2))
  	\ar[rr, "(\ups\comp (\id\pm e))\ten\id"]
  	\ar[dd, "\chi"']
  	\ar[dr, "(\id\pm e)\ten\id"']& &
  (\nf_1(\CE_1')\pm\nf_2(\CE_1'))\ten (\nf_1(\CE_2)\pm\nf_2(\CE_2))
  	\ar[dd, "\chi"]
  \\
  & (\nf_1(\CE_1)\pm\CE)\ten (\nf_1(\CE_2)\pm\nf_2(\CE_2))
  	\ar[ur, "\ups\ten\id"']
  	\ar[dd, "\chi"] &
  \\[-2ex]
  (\nf_1(\CE_1)\ten\nf_1(\CE_2))\pm (\nf_2(\CE_1)\ten\nf_2(\CE_2))
  	\ar[dr, "\id\pm (e\ten\id)"'] & &
  (\nf_1(\CE_1')\ten\nf_1(\CE_2))\pm (\nf_2(\CE_1')\ten\nf_2(\CE_2))
  \\
  & (\nf_1(\CE_1)\ten\nf_1(\CE_2))\pm (\CE\ten\nf_2(\CE_2))
  	\ar[ru,"e''\comp\ups\comp e'"'] &
  \end{tikzcd}
  \end{equation*}
}

\noindent
where, again, $e$ is the one that we inherit from the 
case $f=\ups$, and $e'$ and $e''$ are isomorphisms. The argument for
$f=\id\ten\ups$ is symmetric.

The general case now follows by structural induction on $f$ with the above cases 
serving as the induction base. 

  \item[\eqref{it:coh3}] 
  Let us first argue that w.l.o.g., $f$ and $g$ are of the form 
  $(e_2\pm\id)\comp\ups\comp(\id\pm e_1)$ and $(u_2\pm\id)\comp\ups\comp(\id\pm u_1)$
  with some isomorphism expressions $e_1,e_2,u_1,u_2$. Indeed, by assumption,~$f$ is a composition
  of morphism expressions of the form: $\ups$, $\id\pm e$ and $e\pm\id$ where 
  $e$ ranges over isomorphism expressions. We ensure that $\ups$ occurs at least
  once by using the fact that $\id_{\CA\pm\CB} \equiv (\rho_{\CA}\pm\id_\CB)\comp\ups_{\CA,I,\CB}\comp (\id_{\CA}\pm\gamma_{\CB,I})\comp (\id_{\CA}\pm\rho^\mone_{\CB})$.
  Since $(\id\pm e)\comp\ups\equiv\ups\comp (\id\pm (\id\ten e))$
  and $\ups\comp(e\pm\id)\equiv ((\id\ten e)\pm\id)\comp\ups$, we can rearrange~$f$
  equivalently in such a way that all components of the form $\id\pm e$
  are gathered on the right of the expression and all components of the form $e\pm\id$
  are gathered on the left of the expression. Using the axioms of guarded parameterized 
  monads, we can subsequently replace compositions $\ups\comp\ups$ with a single $\ups$,
  and thus arrive at the requisite form. The same reasoning applies to $g$.
  
  Now the equality $f\equiv g$ follows from the diagram:
  \begin{equation*}
  \begin{tikzcd}[column sep=1em, row sep=2ex]
  & 
  \CA\pm(\CB_1\ten\CD)
  	\rar["\ups"]
  	\ar[dd, "\iso"] &[3em] 
  (\CA\ten\CB_1)\pm\CD
  	\ar[dr, "e_2\pm\id"]
  	\ar[dd, "\iso"] &
  \\               
  \CA\pm\CB
  	\ar[ur, "\id\pm e_1"]
  	\ar[dr, "\id\pm u_1"'] & & & 
  \CC\pm\CD
  \\               
  & \CA\pm(\CB_2\ten\CD)
  	\rar["\ups"] & 
  (\CA\ten\CB_2)\pm\CD
  	\ar[ur, "u_2\pm\id"'] &
  \end{tikzcd}
  \end{equation*}
  using the fact that $\CB_1$ and $\CB_2$ are necessarily isomorphic. That the 
  triangles commute follows from the original Mac Lane's coherence theorem for 
  symmetric monoidal categories~\cite{Mac-Lane71}. 
  \item[\eqref{it:coh4}] Since $\CE$ is assumed to be normal, $\CE=\CE_1\pm\CE_2$,
  where $\CE_1$ and $\CE_2$ do not contain $\pm$. 
  Let us show first, by induction over $\CE_1$, that $\nm(\CE_1)\equiv (\id\pm e_1)\comp\eta$ for some
  isomorphism~$e_1$ between $I$ and a tensor product of some number of copies of $I$. The induction
  base is trivial. For the induction step, let $\CE_1 = \CE_1'\ten\CE_1''$. Then, using the induction hypothesis, $\nm(\CE_1) = \chi\comp(\nm(\CE_1') \ten\nm(\CE_1''))\equiv\chi\comp((\id\pm e_1')\comp\eta\tensor(\id\pm e_1'')\comp\eta)\equiv(\id\pm (e_1'\ten e_1''))\comp\chi\comp(\eta\tensor\eta)\equiv(\id\pm (e_1'\ten e_1'')\comp e)\comp\eta$ where $e\c I\iso I\tensor I$. 
Analogously, we obtain $\nm(\CE_2)\equiv (\id\pm e_2)\comp\eta$. 
  
Now, $\nm(\CE) = \nm(\CE_1\pm\CE_1) = \xi\comp\zeta\comp (\nm(\CE_1) \pm\nm(\CE_2))\equiv
\xi\comp\zeta\comp ((\id\pm e_1)\comp\eta\pm(\id\pm e_2)\comp\eta)\equiv
(\id\ten(e_1\ten(\id\ten e_2)))\comp \xi\comp\zeta\comp (\eta\pm\eta)$. Using the 
axioms of guarded parameterized monads, observe that $\xi\comp\zeta\comp (\eta\pm\eta)$  
is an isomorphism expression, which finishes the proof. \qed
\end{enumerate}
\noqed\end{proof}
It is an open question if a stronger version of the above coherence theorem with 
general $f,g\c\CE_1\to\CE_2$ can be proven. 
In the sequel, we will only deal with guarded parameterized monads over $(\BV,+,0)$.
Recall that a parameterized monad (in the sense of Uustalu~\cite{Uustalu03}) is a bifunctor $T\c\BV\times\BV\to\BV$, 
such that each $T(\argument,X)$ is a monad and each $T(\argument,f)$ is a monad morphism.
Of course, a guarded parameterized monad is meant to be a parameterized monad in this sense. 
This follows from~\autoref{rem:grad} and the following general fact.  
\begin{proposition}\label{prop:gpar-par}
Every $\BV$-graded monad on $(\BV,+,0)$ is a parameterized monad. 
\end{proposition}
\begin{proof}%
A $\BV$-graded monad on $\BV$ is equivalently a lax monoidal functor from $\BV$
to the monoidal category of endofunctors $([\BV,\BV],\comp,\Id)$. Let this lax monoidal functor 
send each $X\in |\BV|$ to $T(\argument, X)\c\BV\to\BV$. Lax monoidal functors preserve monoids, which unravels as follows. Every object $X$
in $\BV$ is a monoid under $\bang\c\iobj\to X$ and $\nabla\c X+X\to X$, and monoids 
in $([\BV,\BV],\comp,\Id)$ are precisely monads. Therefore, $T(\argument, X)$ is a monad for every~$X$. 
Likewise, every morphism $f\c X\to Y$ in $\BV$ is a monoid morphism, and hence induces 
a monoid morphism from $T(\argument, X)$ to $T(\argument, Y)$, i.e.\ a monad morphism.
\end{proof}
Explicitly, for a guarded parameterized monad $\pm$ we obtain parameterized  
multiplication transformation:
\begin{align*}
\mu_{X,Y} =\bigl( (X\pm Y)\pm Y&\;\xto{\xi} X\pm (Y+Y)\xto{\id\pm\nabla} X\pm Y\bigr)\\[-6ex]
\end{align*}
\begin{theorem}\label{thm:grep}
Given co-Cartesian $\BV$ and an identity-on-object functor 
$J\c\BV\to\BC$ strictly preserving coproducts, $\BC$ is guarded and 
$\GHom*{\BC}$ is representable iff\/
  $\BC\iso\BV_{\argument\pm 0}$ for a guarded parameterized monad $(\pm,\eta,\ups,\chi,\xi,\zeta)$,
  the compositions 
  $\ups_{X,Y,\iobj}\comp (\id\pm\inl)$
  are all monic and $f\c X\to Y\gsep Z$ iff $f$ factors through $Y\pm (Z+0) \xto{\ups} (Y+ Z)\pm\iobj$.
\end{theorem}
\begin{proof}
($\Rightarrow$) Let $\GHom*{\BC}$ be $J$-representable. By \autoref{thm:gpm}, 
assume w.l.o.g.\ that $\BC=\BV_{\BBT}$ where $T=\argument\pm 0$, and let $\eps$
and $(\argument)^\sharp$ be the corresponding structure, belonging to $\pm$.
Let $\eta$ be the unit of~$\BBT$. We obtain the remaining transformations~$\ups,\xi,\zeta$ 
and $\chi$ by universality as follows:
\begin{align*}
\ups  =&\; \bigl(X\pm (Y + Z)  \xto{\eps} (X+(Y+Z))\pm\iobj\iso ((X+Y)+Z))\pm\iobj\bigr)^\sharp,\\
\xi   =&\; \bigl((X\pm Y) \pm Z  \xto{\eps} (X\pm Y+Z)\pm\iobj\xto{[T(\id+\inl)\comp\eps,\eta\comp\inr\comp\inr]^\klstar}  (X+(Y+Z))\pm\iobj\bigr)^\sharp,\\
\zeta =&\; \bigl(X\pm (Y\pm Z)  \xto{\eps} (X+Y\pm Z)\pm\iobj\xto{[\eta\comp\inl, T\inr\comp\eps]^\klstar}  (X+(Y+Z))\pm\iobj\bigr)^\sharp,\\
\chi  =&\; \bigl(X\pm Y + Z\pm V  \xto{\eps+\eps} (X+Y)\pm\iobj + (Z+V)\pm\iobj\\*&\hspace{6.85em}\xto{T[\inl+\inl,\inr+\inr]} ((X+Y)+(Z+V))\pm\iobj\bigr)^\sharp.
\end{align*}
It is clear by definition that $f^\sharp$ is mono as long as $f$ is mono, hence 
$\ups$ is mono. The characterization of the guardedness predicate follows from 
\autoref{thm:gpm}. The laws of guarded parameterized monad all follow by postcomposition 
with $\eps$ and using the fact that it is mono.

($\Leftarrow$) Let $\BC=\BV_{\argument\pm 0}$ for a guarded parameterized monad $(\pm,\eta,\ups,\chi,\xi,\zeta)$.
We define $\eps\c X\pm Y\to (X+Y)\pm 0$ as $X\pm Y\iso X\pm (Y+0) \xto{\ups} (X+Y)\pm 0$,
which is monic, since~$\ups$ is so by assumption. This yields a unique~$f^\sharp$
for every guarded~$f$, by definition of the guardedness predicate. The only non-trivial condition
of \autoref{thm:gpm}, which is left to verify,	 is that the guardedness predicate is well-defined.
\begin{itemize}[wide]
  \item \textbf{(trv$_\gtag$)} Given $f\c X\to Y$, $\eta\comp\inl\comp f = \eps\comp (\id\pm\bang)\comp\eta\comp f$
  is thus guarded.
  \item \textbf{(par$_\gtag$)} Given $f\c X\to \IB{V}{W}$ and $g\c Y\to \IB{V}{W}$,
  $[\eps\comp f,\eps\comp g] = \eps\comp [f,g]$ is guarded.
  \item \textbf{(cmp$_\gtag$)} Given $f\c X\to \IB{Y}{Z}$, $g\c Y\to \IB{V}{W}$ and $h\c Z\to \IB{(V+W)}{\iobj}$,
  observe first that
  \begin{align*}
  [\eps\comp g,h]^\klstar\comp\eps\comp f 
  &\; = \mu\comp (\IB{[\eps\comp g,h]}{\id})\comp\eps\comp f\\
  &\; = (\id\pm\nabla)\comp\xi\comp (\IB{[\eps\comp g,h]}{\id})\comp\eps\comp f\\
  &\; = (\id\pm\nabla)\comp\xi\comp (\IB{[\eps,\id]}{\id})\comp\eps\comp(\IB{g}{h})\comp f.
  \end{align*}  
  That is, we are left to show that $(\id\pm\nabla)\comp\xi\comp (\IB{[\eps,\id]}{\id})\comp\eps$
  factors through $\eps$. Observe that $(\nabla\pm\nabla)\comp\chi = \nabla$. Therefore,
  \begin{align*}
  (\id\pm&\nabla)\comp\xi\comp (\nabla\pm\id)\comp (\IB{(\eps+\id)}{\id})\comp\eps\\
  &\; = (\id\pm\nabla)\comp\xi\comp ((\nabla\pm\nabla)\comp\chi\pm\id)\comp (\IB{(\eps+\id)}{\id})\comp\eps\\
  &\; = (\id\pm\nabla)\comp (\nabla\pm(\nabla+\id))\comp\xi\comp(\chi\pm\id)\comp (\IB{(\eps+\id)}{\id})\comp\eps\\
  &\; = (\nabla\pm \nabla\comp(\nabla+\id))\comp\xi\comp(\chi\pm\id)\comp (\IB{(\eps+\id)}{\id})\comp\eps\\
  &\; = ((\nabla+\nabla)\pm \nabla\comp(\nabla+\id))\comp ([\inl+\inl,\inr+\inr]\pm\id)\comp\xi\comp (\IB{\chi\comp(\eps+\id)}{\id})\comp\eps.
  \end{align*}
  Using coherence, $([\inl+\inl,\inr+\inr]\pm\id)\comp\xi\comp (\IB{\chi\comp(\eps+\id)}{\id})\comp\eps\c (\IB{V}{W})\pm ((V + W)\pm\iobj)\to
  \IB{((V+V)+(W+W))}{((\iobj+\iobj)+\iobj)}$ can be factored through~$\ups$, and hence 
  the entire expression factors through~$\eps$.
  \qed
\end{itemize}
\noqed
\end{proof}
Vacuous guardedness is clearly representable and by \autoref{thm:grep} corresponds 
to those guarded parameterized monads $\pm$, which do not depend on the parameter,
i.e.\ to monads.
\begin{example}
Let us revisit \autoref{exa:resump}. Let $\IB{X}{Y} = T(X+HT_{H}(X+Y))$, and 
note that~$\IB{\argument}{\iobj}$ is isomorphic to $T_H$. 
Assuming the existence of some morphism~${p\c 1\to H1}$,   
for every $X$, we obtain the final map $\hat p\c 1\to T_HX$, induced by the 
coalgebra map $1\xto{\eta\comp\inr\comp p} T(X + H1)$. Now, $T(\inl+\id)$ is a section,
since $T[\id+H\hat p\comp p\comp\bang\comma\inr]\comp T(\inl+\id)$ is the identity.  
By \autoref{thm:gpm},~$\IB{}{}$ is a guarded parameterized monad.
\end{example}
\begin{example}
Let us revisit \autoref{exa:guard}. %
Let 
$\IB{X}{Y} =\realp\times X+\realsp\times Y+\realpe$.
Then $\IB{X}{\iobj}\iso \realp\times X+\realpe$ and there is an obvious injection
$\eps_{X,Y}$ from $\IB{X}{Y}$ to $\IB{(X+Y)}{\iobj}$. By definition, every guarded 
$f\c X\to \IB{Y}{Z}$ uniquely factors through $\eps_{Y,Z}$, and hence $\IB{}{}$
is a guarded parameterized monad.
\end{example}
\begin{definition}[Strong Guarded Parameterized Monad]\label{def:gspar}
A guarded parameterized monad $(\pm,\eta\comma\ups\comma\chi\comma\xi\comma\zeta)$ 
is strong, if\/ $\pm$ is strong as a monad in the first argument and as a functor in the second 
argument, and the diagram
\begin{equation*}
\begin{tikzcd}[column sep=1.75em, row sep=normal]
X\times (Y\pm Z)
  \rar["\id\times\eps"]
  \dar["\tau"'] &[.5em] 
X\times (Y+Z)\pm\iobj
  \rar["\tau"] &[-.75em] 
(X\times (Y+Z))\pm0
  \rar["\dist\pm 0"] &[.75em] 
(X\times Y + X\times Z)\pm 0
  \dar["(\id+\snd)\pm 0"]\\                
(X\times Y)\pm Z
  \ar[rrr,"\eps"] &&& (X\times Y + Z)\pm\iobj
\end{tikzcd}
\end{equation*}
commutes, where $\eps_{X,Y} = \ups_{X,Y,\iobj}\comp (\id\pm\inl)$ and $\tau$ is the monadic strength of\/ $\pm$.
\end{definition}
\begin{remark}
Strength is commonly referred to as a ``technical condition''. This is 
justified by the fact that in self-enriched categories, strength is equivalent to 
enrichment of the corresponding functor or a monad~\cite{Kock72}, and in foundational 
categories, like $\Set$, every functor and every natural transformation are canonically enriched
w.r.t.\ Cartesian closeness as the self-enrichment structure. Then the canonical strength 
$\rho_{X,Y}\c X\times FY\to F(X\times Y)$ for a functor~$F$ is defined by the expression 
$\rho_{X,Y} = \lambda (x,z).\,F(\lambda y.\,(x,y))(z)$. We conjecture that the
strengths involved in \autoref{def:gspar} are technical in the same sense, in particular,
the requested commutative diagram is entailed by enrichment of $\eps$.\sgnote{Prove.}
\end{remark}
Finally, let us establish the analogue of \autoref{thm:grep} for Freyd 
categories.
\begin{theorem}\label{thm:frey-rep} 
A Freyd category $(\BV,\BC,J(\argument),\oslash)$ is guarded 
and $\GHom*{\BC}$ is representable iff $\BC\iso\BV_{\argument\pm 0}$ for a strong 
guarded parameterized monad $(\pm,\eta,\ups,\chi,\xi,\zeta)$, the compositions 
$\ups_{X,Y,\iobj}\comp (\id\pm\inl)$
  are all monic and $f\c X\to Y\gsep Z$ iff $f$ factors through $Y\pm (Z+0) \xto{\ups} (Y+ Z)\pm\iobj$.
\end{theorem}
\begin{proof}
\autoref{thm:grep} and \autoref{prop:freyd} jointly imply that 
$\BC$ is a representable guarded Freyd category~iff 
\begin{itemize}[wide]
  \item $\BC\iso\BV_{\argument\pm 0}$ for a guarded parameterized monad $(\pm,\eta,\ups,\chi,\xi,\zeta)$,
  \item $\ups$ is componentwise monic,
  \item $f\c X\to Y\gsep Z$ iff $f$ factors through $Y\pm (Z+0) \xto{\ups} (Y+ Z)\pm\iobj$,
  \item $T=(\argument)\pm\iobj$ is strong and $(T\dist)\comp\tau\comp (\id\times\eps)\c X\times (Y\pm Z)\to T(X\times Y+X\times Z)$
  uniquely factors through $\eps\c X\times Y\pm X\times Z\to T(X\times Y+X\times Z)$ where 
  $\tau$ is the strength of~$\BBT$.  
\end{itemize}
This yield strength for both sides of $\pm$ by 
  composition:
  \begin{align*}
  X\times (Y\pm Z)&\;\to (X\times Y)\pm (X\times Z)\xto{\id\pm\snd} (X\times Y)\pm Z,\\*
  X\times (Y\pm Z)&\;\to (X\times Y)\pm (X\times Z)\xto{\snd\pm\id} Y\pm (X\times Z).
  \end{align*} 
The diagram in \autoref{def:gspar} is thus satisfied by definition. The 
axioms of strength are checked routinely.
\end{proof}
For a strong guarded parameterized monad $\pm$, let $\wave\tau$ be the composition
\begin{align*}
X\times (Y\pm Z)\xto{\Delta\times\id}&\; (X\times X)\times (Y\pm Z)\iso X\times (X\times (Y\pm Z))\\* 
\xto{\id\times\rho} &\; X\times (Y\pm (X\times Z)) \xto{\tau} (X\times Y)\pm (X\times Z)
\end{align*}
where $\tau$ is the monadic strength of $\pm$ and $\rho$ is the functorial strength 
of $\pm$.
It is easy to check that~$\tau$ and $\rho$ are derivable from $\wave\tau$, and 
in the sequel, we will include it as the 
last element in a tuple $(\pm,\eta,\ups,\chi,\xi,\zeta,\wave\tau)$, defining a strong 
guarded parameterized monad.
\begin{figure*}[t]
\begin{gather*}
  \frac{}{\sem{x_1\c A_1,\ldots,x_n\c A_n\ctxv x_i\c A_i} = \proj_i}\\[2ex]  
  \frac{h=\sem{\G\ctxv v\c A}}{\sem{\G\ctxv f(v)\c B} = \sem{f}\comp h}\qquad 
  \frac{h=\sem{\G\ctxv v\c A}}{\sem{\G\ctxc f(v)\c B} = \sem{f}\comp h}\\[2ex]  
  \frac{h=\sem{\G\ctxv v\c A}}{\sem{\G\ctxc \ret v\c A\gsep B} = \eta\comp h}\\[2ex]
  \frac{h = \sem{\G\ctxc p\c A\gsep B}\qquad h_1=\sem{\G,x\c A\ctxc q\c C\gsep D}\qquad h_2 =\sem{\G,y\c B\ctxc r\c (C+D)\gsep\iobj}}{
  \begin{array}{r@{}l}
  	\sem{\G\ctxc\docase{p}{\linj\,& x\mto q}{\rinj y\mto r}\c C\gsep D} \\[1ex]
  	=\;& (\nabla\pm\id)\comp\ups\comp\mu\comp ((\id\pm\inr)\pm[\id,\bang]) \comp\zeta\comp  (h_1\pm h_2) \comp\wave\tau\comp \brks{\id, h}
  \end{array}}\\[2ex]
  \frac{}{\sem{\G\ctxc\oname{init} v\c A} = \bang}\qquad
  \frac{h=\sem{\G\ctxv v\c A}}{\sem{\G\ctxv\linj v\c A+B} = \inl\comp h}\qquad
  \frac{h=\sem{\G\ctxv v\c B}}{\sem{\G\ctxv\rinj v\c A+B} = \inr\comp h}\\[2ex]
  \frac{h=\sem{\G\ctxv v\c A+B}\qquad h_1=\sem{\G,x\c A\ctxc p\c C\gsep D}\qquad h_2=\sem{\G,y\c B\ctxc q\c C\gsep D}}{\sem{\G\ctxc\case{v}{\linj x\mto p}{\rinj y\mto q}\c C\gsep D} 
  = (\nabla\pm\nabla) \comp\chi \comp(h_1+h_2)\comp\dist\comp \brks{\id, h}}\\[2ex]
  \frac{h_1=\sem{\G\ctxv v\c A}\qquad h_2=\sem{\G\ctxv w\c B}}{\sem{\G\ctxv \brks{v,w}\c A\times B} = \brks{h_1,h_2}}
\end{gather*}
 \caption{Denotational semantics of guarded FGCBV over guarded parameterized monads.}
  \label{fig:mon-sem}
\end{figure*}

Finally, we can interpret the guarded version of FGCBV over a strong 
guarded parameterized monad $(\pm,\eta,\ups,\chi,\xi,\zeta,\wave\tau)$ on $\BV$. 
A semantics of~$(\Sigv,\Sigc)$ then assigns
\begin{itemize}[wide]
  \item an object~$\sem{A}\in |\BV|$ to each sort~$A$;
  \item a morphism~$\sem{f}\in\BV(\sem{A},\sem{B})$ to each~$f\c A\to B\in\Sigv$;
  \item a morphism~$\sem{f}\in\BV(\sem{A},\sem{B}\pm\sem{C})$ to each~$f\c A\to B\gsep C\in\Sigc$;
\end{itemize}
This semantics extends to types as before and to terms in context with the assignments 
in \autoref{fig:mon-sem}. Let us spell out the most sophisticated morphism 
corresponding to the rule for~$\oname{docase}$: 
\begin{align*}
\sem{&\G}
	\xto{\brks{\id, h}} 													  \sem{\G}\times(\sem{A}\pm\sem{B})
	\xto{\wave\tau_{\sem{\Gamma},\sem{A},\sem{B}}}  (\sem{\G}\times\sem{A})\pm(\sem{\G}\times\sem{B})\\
	&\xto{h_1\pm h_2} 														  (\sem{C}\pm\sem{D})\pm((\sem{C}+\sem{D})\pm\iobj)
	\xto{\zeta_{\sem C\pm\sem D,\sem{C}+\sem{D},\iobj}}     (\sem{C}\pm\sem{D})\pm((\sem{C}+\sem{D})+\iobj)\\
	&\xto{(\id\pm\inr)\pm [\id,\bang]}     (\sem{C}\pm(\sem{C}+\sem{D}))\pm(\sem{C}+\sem{D})
	\xto{\mu_{\sem{C},\sem{C}+\sem{D}}}     \sem{C}\pm(\sem{C}+\sem{D})\\*
	&	\xto{\ups_{\sem{C},\sem{C},\sem{D}}} 						(\sem{C}+\sem{C})\pm\sem{D} \xto{\nabla\pm\id}
\sem{C}\pm\sem{D}
\end{align*}
Note that what allows us to sidestep the monicity condition of the representability 
criterion (\autoref{thm:grep}) is that we gave up on the assumption that the space
of guarded morphisms $X\to Y\pm Z$ injectively embeds into the space of all morphisms 
$X\to (Y + Z)\pm\iobj$, in particular, the entire notion of guardedness predicate 
is eliminated.

\section{Conclusions and Further Work}\label{sec:conc}
We investigated a combination of FGCBV and guardedness, drawing inspiration
from previous work relating Freyd categories to strong monads via a natural
representability condition for certain presheaves. An abstract notion of
guardedness naturally fits into the FGCBV paradigm and gives rise to more
general formats of presheaves, which must be representable, e.g., to interpret
higher-order (guarded) functions. In our case, the representability requirement
gave rise to a novel categorical structure — we dub it a (strong) guarded
parameterized monad — that encapsulates the computational effects under
consideration while providing guardedness guarantees.

We regard our present results as a prerequisite step for implementing guarded
programs in existing higher-order languages, such as Haskell, and in proof
assistants with strict support of the propositions-as-types discipline,
such as Coq and Agda, where unproductive recursive definitions cannot be
implemented directly, and thus guarded iteration is particularly significant.
It would be interesting to further refine guarded parameterized monads to
include quantitative information about how productive a computation is, or how
unproductive it is, so that this relative unproductivity could possibly be
cancelled out by composition with something very productive. Another strand for
future work arises from the observation that guarded iteration is a formal dual
of guarded recursion~\cite{GoncharovSchroder18}.

A good deal of the present theory can be easily dualized, which will presumably 
lead to guarded parameterized comonads and comonadic recursion --  we are planning
to investigate these structures from the perspective of comonadic notion of 
computation~\cite{UustaluVene08}. In terms of syntax, a natural extension of 
fine-gain call-by-value is call-by-push-value~\cite{Levy99}. We expect it to be 
a natural environment for analyzing the above-mentioned aspects in the style of 
the presented approach.

As we demonstrated, guarded parameterized monads emerge as an answer to 
a very natural representability question, but the resulting notion, i.e.\ 
\autoref{def:pm}, admittedly appears to be rather unwieldy. It involves five natural 
transformations, two of which ($\eta$ and~$\xi$) render guarded parameterized monads
as graded~\cite{FujiiKatsumataEtAl16} or parametric monads~\cite{Mellies17}.
Each of the remaining transformations has its specific role in governing guardedness
guarantees. They ensure that a guardedness guarantee can be weakened ($\ups$),
that independent guardedness guarantees can be merged ($\chi$), and that nested 
guardedness guarantees can be flattened ($\zeta$). Numerous coherence conditions 
between these transformations are vital for the coherence theorem, and it seems 
that not much can be done to simplify them significantly. One 
seemingly natural idea is to replace the natural transformation 
$\chi\c A\pm B \ten C\pm D\to (A \ten C)\pm (B \ten D)$ with a more elementary 
transformation $\kappa\c A\ten B \pm C\to (A \ten B)\pm C$ from which $\chi$
can be derived. This, however, does not have a straightforward simplifying effect on the 
coherence conditions, in particular on the condition describing the interaction of $\chi$ 
and associativity transformations. Outside the context of coherence, the only useful example 
of $\ten$ known presently is the binary coproduct functor $+$, and the only useful 
candidate for $\chi$ in this case is $[\inl\pm\inl,\inr\pm\inr]$.
In this special case \autoref{def:pm} might be possible to simplify. 

\section*{Acknowledgement}
The author would like to thank anonymous reviewers of the present 
and previous editions of the paper for their diligence in their effort to improve it.

\bibliographystyle{alphaurl}
\bibliography{refs}

\end{document}